\documentclass[12pt]{article}
\pdfoutput=1
\usepackage{mathtools,mathrsfs}
\usepackage{amsmath, amsfonts,amsthm,mathabx}
\usepackage{graphicx, xcolor}
\usepackage{enumerate}
\usepackage{tikz}
\usepackage{hhline}
\usepackage{multirow,array,algcompatible}

\usepackage{booktabs}
\newcommand{\ra}[1]{\renewcommand{\arraystretch}{#1}}
\newcommand{\m}[1]{{\mathsf #1}}

\DeclareMathOperator{\Tr}{Tr}
\newcommand{\R}{\mathbb{R}}

\newcommand{\eg}{\emph{e.g.}}
\newcommand{\ie}{\emph{i.e.}}

\newcommand{\etal}{et al.\@ }

\renewcommand{\O}{O}
\newcommand{\softO}{{\tilde O}}

\newcommand{\rd}{\mathcal{R}}
\newcommand{\sk}{\mathcal{S}}
\newcommand{\I}{\mathcal{I}}

\newcommand{\N}{\mathcal{N}}
\newcommand{\F}{\mathcal{F}}

\renewcommand{\L}{\mathscr{L}}

\renewcommand{\H}{\mathcal{H}}
\newcommand{\thetahat}{\hat{\theta}_{\text{MLE}}}

\newtheorem{definition}{Definition}

\newtheorem{remark}{Remark}
\newtheorem{theorem}{Theorem}

\usetikzlibrary{decorations.pathreplacing, patterns}
\colorlet{lightgray}{black!20}
\colorlet{darkgray}{black!50}

\newcommand{\revision}[1]{#1}

\usepackage{algorithm, algpseudocode}
\usepackage{url}
\usepackage{hyperref,fullpage}
\usepackage[capitalize]{cleveref}

\renewcommand{\L}{\mathscr{L}}

\newcommand{\peelidx}[2] {
\I_{#1;#2}}
\newcommand{\G}[2]{
	\m{G}_{#1;#2}
}

\newcommand{\Ghat}[3]{
	\widehat{\m{G}}_{#1;#2	#3}
}
\newcommand{\W}[2]{
	\m{W}_{#1;#2}
}
\newcommand{\U}[2]{
	\m{U}_{#1;#2}
}
\newcommand{\M}[2]{
	\m{M}_{#1;#2}
}

\newcommand{\nn}[2]{
	n_{#1;#2}
}

	\newcommand{\funding}[1]{{\bf Funding:} #1}

	\newcommand{\email}[1]{\url{#1}}
	\renewcommand{\O}{O}

\ra{1.3}

\title{Fast spatial Gaussian process\\ maximum likelihood estimation\\ via skeletonization factorizations}

\author{
	Victor Minden%
	\thanks{Institute for Computational and Mathematical Engineering, Stanford University, Stanford, CA 94305 (\email{vminden@stanford.edu}). \funding{Stanford Graduate Fellowship in Science \& Engineering and U.S. Department of Energy Computational Science
Graduate Fellowship (grant number DE-FG02-97ER25308).}}%
\and
	Anil Damle%
	\thanks{\revision{Department of Computer Science, Cornell University, Ithaca, NY 14850 (\email{damle@cornell.edu})}. \funding{Simons Graduate Research Assistantship and National Science Foundation Graduate Research Fellowship (grant number DGE-1147470).}}%
	\and
	 Kenneth L. Ho%
  \thanks{Department of Mathematics, Stanford University, Stanford, CA 94305.  Current address: TSMC Technology Inc., 2851 Junction Ave., San Jose, CA 95134. (\email{klho@alumni.caltech.edu}). \funding{National Science Foundation Mathematical Sciences Postdoctoral Research Fellowship (grant number DMS-1203554).}}%
  \and
	Lexing Ying\thanks{Department of Mathematics and Institute for Computational and Mathematical Engineering, Stanford University, Stanford, CA 94305 (\email{lexing@math.stanford.edu}). \funding{National Science Foundation (grant number DMS-1328230 and DMS-1521830) and U.S. Department of Energy Advanced Scientific Computing Research program (grant number DE-FC02-13ER26134 and DE-SC0009409).}}
}

\date{}

\begin{document}

\maketitle

\begin{abstract}
Maximum likelihood estimation for parameter-fitting given observations from a Gaussian process in space is a computationally-demanding task that restricts the use of such methods to moderately-sized datasets.  We present a framework for unstructured observations in two spatial dimensions that allows for evaluation of the log-likelihood and its gradient (\ie, the score equations) in $\softO(n^{3/2})$ time under certain assumptions, where $n$ is the number of observations.  Our method relies on the skeletonization procedure described by Martinsson \& Rokhlin \cite{martinsson-rokhlin} in the form of the recursive skeletonization factorization of Ho \& Ying \cite{hifie}.  Combining this with an adaptation of the matrix peeling algorithm of Lin \etal \cite{Lin2011} for constructing $\mathcal{H}$-matrix representations of black-box operators, we obtain a framework that can be used in the context of any first-order optimization routine to quickly and accurately compute maximum-likelihood estimates.
\end{abstract}

\section{Introduction}
Gaussian processes are commonly used in the applied sciences as a statistical model for spatially-indexed observations.  In such applications, each observation $z_i\in\R$ of some quantity is associated with a corresponding location $x_i\in\Omega\subset\R^d$ with $d=2$ or $3$.  Given a prescribed covariance kernel $K(\cdot,\cdot;\theta)$ that maps $\R^d\times\R^d$ to $\R$ and is specified up to some parameter vector $\theta\in\R^p$, any vector of observations $z=[z_1,\dots,z_n] \in \R^n$ (with associated locations $\{x_i\}_{i=1}^n$) is assumed to be randomly distributed as a multivariate Gaussian 
\begin{equation}\label{eq:normaldist}
z \sim N(0, \m{\Sigma}(\theta)) \text{ with covariances } [\m{\Sigma}(\theta)]_{ij} = K(x_i,x_j;\theta).
\end{equation} 
We assume the mean of the process to be 0 (\ie, known) for simplicity {though, as we discuss later, this is not stricly necessary.}

{Neglecting $\theta$, common choices for the covariance kernel include the {family of rational quadratic kernels} 
\begin{align}\label{eq:rationalquadratic}
K(x,y) = \left(1 + \frac{\|x-y\|^2}{2\alpha}\right)^{-\alpha},
\end{align}
{which has corresponding processes that are infinitely differentiable in the mean-squared sense for all $\alpha > 0$.  Notably, this family includes the Gaussian or squared exponential kernel as a limiting case as $\alpha\to\infty$.  Another popular family is the Mat\'ern family of kernels}
\begin{align}\label{eq:matern}
K(x,y) = \frac{1}{\Gamma(\nu)2^{\nu-1}}\left(\sqrt{2\nu}\cdot\|x-y\|\right)^\nu K_\nu\left(\sqrt{2\nu}\cdot\|x-y\|\right),
\end{align}
where $K_\nu(\cdot)$ is the modified second-kind Bessel function {of order $\nu$,} $\Gamma(\cdot)$ is the gamma function, and the corresponding process is $\lfloor{\nu}\rfloor$ times mean-squared differentiable.   To explicitly parameterize $K(\cdot,\cdot;\theta)$, we might introduce a {correlation length parameter $\rho$} leading to
\begin{align*}
K(x,y;\theta) &=  \left(1 + \frac{\|x-y\|^2}{2\alpha\rho^2}\right)^{-\alpha}
\end{align*}
in the case of the rational quadratic kernel.  The fundamental parameters of the kernel family  e.g., $\alpha$ or $\nu$, may be considered as part of $\theta$ or fixed $\emph{a priori}$. }

Typically, the parameter vector $\theta$ is unknown and must be estimated from the data.  For example, given a parameterized family of kernels and a set of observations, we might want to infer $\theta$ for later use in estimating the value of the field at other spatial locations as in kriging (see Stein \cite{stein1999interpolation}).  In this paper we consider the general Gaussian process maximum likelihood estimation (MLE) problem for $\theta$: given an observation vector $z\in\R^n$, find $\thetahat$ maximizing the Gaussian process log-likelihood
\begin{align}\label{eq:loglikelihood}
\ell(\theta) &\equiv -\frac{1}{2}z^T\m{\Sigma}^{-1}z - \frac{1}{2}\log | \m{\Sigma} | - \frac{n}{2}\log 2\pi,
\end{align}
where we have dropped the explicit dependence of $\m{\Sigma}$ on $\theta$ for notational convenience.

If $\theta$ is unconstrained, then $\thetahat$ is given by maximizing \eqref{eq:loglikelihood} over all of $\R^p$.  In this case, it is possible under certain assumptions to obtain the maximum likelihood estimate by solving the score equations $g(\theta) = 0$, where the gradient of the log-likelihood $g(\theta)\in\R^p$ is given component-wise by
\begin{align}\label{eq:gradient}
g_i&\equiv\frac{\partial \ell(\theta)}{\partial \theta_i} = \frac{1}{2}z^T\m{\Sigma}^{-1}\m{\Sigma}_i\m{\Sigma}^{-1}z - \frac{1}{2}\Tr(\m{\Sigma}^{-1}\m{\Sigma}_i),\quad i=1,\dots,p,
\end{align}
where $\m{\Sigma}_i\equiv \frac{\partial}{\partial \theta_i}\m{\Sigma}$.
This may be accomplished without evaluating $\ell(\theta)$ as is done by, \eg, Anitescu \etal \cite{Anitescu2012} and Stein \etal \cite{Stein2013}.  In contrast, we consider in this paper the use of first-order methods for nonlinear optimization that, at each iteration, use both gradient and log-likelihood evaluations to find a local optimum.  This allows for the treatment of constraints if desired, though we note the methods here are equally applicable to the unconstrained case.

For any given $\theta$, both $\ell(\theta)$ and the gradient $g(\theta)$ contain a number of terms whose evaluation is traditionally computationally expensive.  For example, the Cholesky decomposition of $\m{\Sigma}$ may be used to calculate $z^T\m{\Sigma}^{-1}z$ and $z^T\m{\Sigma}^{-1}\m{\Sigma}_i\m{\Sigma}^{-1}z$ as well as the log-determinant $\log |\m{\Sigma}|$ and trace $\Tr(\m{\Sigma}^{-1}\m{\Sigma}_i)$, but the asymptotic computation and storage complexities are $\O(n^3)$ and $\O(n^2)$ respectively.  This is prohibitively expensive for datasets with a large number of observations, necessitating alternative approaches.

\subsection{Our method}
The contribution of this paper is a framework for efficiently finding $\thetahat$ by taking advantage of fast hierarchical matrix algorithms developed in the numerical linear algebra community.  Such algorithms exploit the fact that linear operators defined in terms of pairwise kernel evaluations between points embedded in $\R^d$ frequently exhibit \emph{hierarchical rank structure}, as we discuss in \cref{sec:factor} (briefly, many different-sized off-diagonal blocks of the matrix are close to low-rank and thus compressible).  Our framework has two key parts:
\begin{enumerate}[(I)]
	\item Construct a fast approximate hierarchical factorization of the covariance matrix $\m{\Sigma}$ and its derivatives $\m{\Sigma}_i$ for $i=1,\dots,p$ in a matrix-free fashion using the kernel function $K(\cdot,\cdot;\theta)$ and the points $\{x_i\}_{i=1}^n$.
	\item Additionally factor the derivatives of the covariance matrix $\m{\Sigma}_i$ for $i=1,\dots,p$ through the same approach as in (I).  Use the hierarchical factorizations as a fast black-box operator for approximately applying $\m{\Sigma}^{-1}\m{\Sigma}_i$ for $i=1,\dots,p$, and use this operator and the points $\{x_i\}_{i=1}^n$ to compute the traces $\Tr(\m{\Sigma}^{-1}\m{\Sigma}_i)$ through a randomized ``matrix peeling'' scheme for efficiently extracting the trace of a hierarchically rank-structured linear operator.
\end{enumerate}
	In both parts, the approximation accuracy is well-controlled by specified tolerance parameters intrinsic to the algorithms.

	For any $\theta$, the approximate factorization of part (I) can be used to efficiently evaluate {the terms composing $\ell(\theta)$ (including the log-determinant}), which overlaps with recent work by Ambikasaran \etal \cite{ambikasaran} that addresses the use of hierarchical factorizations for kernelized covariance matrices {for computing these terms} (see also Khoromskij \etal \cite{khoromskij2008} and B\"orm \& Garcke \cite{borm2007} for earlier work on $\mathcal{H}$-matrix techniques for fast computation of matrix-vector products with $\m{\Sigma}$ in a Gaussian process context).  This piece of the framework alone gives sufficient machinery to perform black-box optimization using numerical derivatives.  However, {using finite differences of an approximate log-likelihood can magnify approximation errors, which can lead to larger inaccuracies in the approximate gradient depending on, \eg, the conditioning of $\m{\Sigma}$} (see \cref{sec:ocean1} for a relatively benign example).  Therefore, central to our framework is the computation in (II) of the gradient components $g_i$ for $i=1,\dots,p$, which requires the trace terms $\Tr(\m{\Sigma}^{-1}\m{\Sigma}_i)$.  

	In the simplest form detailed in this paper, our framework employs the {\emph{recursive skeletonization factorization}} \cite{hifie} as the approximate hierarchical factorization variant of choice for $\m{\Sigma}$ and $\m{\Sigma}_i$, $i=1,\dots,p$.  We then compute the trace terms in the gradient using an adaptation of the matrix peeling algorithm of Lin \etal \cite{Lin2011}.  Combining these two tools, we obtain an efficient method for evaluating \eqref{eq:loglikelihood} and \eqref{eq:gradient}---and, ultimately, finding $\thetahat$---with high and controllable accuracy using a black-box first-order optimization package (\eg, \texttt{fminunc} or \texttt{fmincon} in MATLAB\textsuperscript{\textregistered}).

While the framework of this paper technically applies to observations in $d$ dimensions for general $d$, the computational complexity increases in high dimensions due to how the runtime of hierarchically rank-structured factorizations depends on the numerical rank of off-diagonal blocks.  For example, applying these methods in the case $d=1$ is essentially optimal in the sense that off-diagonal matrix blocks have numerical rank that is not strongly dependent on the number of observations.  We direct the reader to Ambikasaran \etal \cite{ambikasaran} for extensive numerical examples of factoring kernel matrices in this case.  In contrast, for $d=3$ the observed rank growth is in general much larger and leads to greater asymptotic complexities.  We focus on the case $d=2$ in the remainder of this paper, but note the broader applicability.

\subsection{Alternative approaches}
Due to the prevalance of Gaussian process models, a number of methods exist in the literature for fast computations involving kernelized covariance matrices.  To decrease apply, solve, and storage costs, $\m{\Sigma}$ can be replaced with a sparser ``tapered'' approximant as described by Furrer \etal \cite{nychka}, wherein the desired covariance kernel is multiplied pointwise with a compactly-supported tapering function to attain sparsity.  Of course, the computational benefit of tapering depends on the sparsity of the resulting approximant, which is limited by the desired accuracy if the correlation length of the kernel is not small.

If $\m{\Sigma}$ decomposes naturally into the sum of a diagonal and a numerically low-rank matrix then such a decomposition can be quite efficient for computation (see Cressie \& Johannesson \cite{cressie}), but this representation is too simple for the applications and kernel functions we consider.  Extending this to a general sparse-plus-low-rank model by replacing the diagonal piece with a tapered covariance kernel can perform better than either a tapered or low-rank $\m{\Sigma}$ alone, as shown by Sang \& Huang \cite{Sang2012} (see also Vanhatal \etal \cite{vanhatalo}).

For cases where the underlying process is stationary and the observations lie on a regular grid it is possible to directly approximate the log-likelihood using spectral approximations due to Whittle \cite{whittle} or to quickly apply the covariance matrix in Fourier space to solve linear systems with an iterative method.  Further, in such cases these systems can be preconditioned using the method of Stein \etal \cite{steinpc} yielding efficient methods for many important problem classes.  For irregularly spaced data such as we consider in this paper, however, these approaches do not apply directly.  One approach for log-likelihood evaluation (and thus derivative-free optimization) with generally distributed data is that of Aune \etal \cite{aune}, which offers an involved framework for approximating the log-determinant using Krylov methods, assuming the covariance matrix or its inverse can be efficiently applied.  More recently, Castrill\'on-Cand\'as \etal \cite{castrillon} demonstrate a combination of multi-level preconditioning and tapering for fast derivative-free restricted maximum likelihood estimation, though gradient computation is not discussed.

An alternative approach to approximating the Gaussian process log-likelihood directly is to explicitly construct and solve a different set of estimating equations that is less computationally cumbersome.  For example, the Hutchinson-like sample average approximation (SAA) estimator \cite{Anitescu2012,Stein2013} falls into this category, as do the composite likelihood methods described by, \eg, Vecchia \cite{vecchia} and Stein \etal \cite{steincomp} and their extension, block composite likelihood methods (see, \eg, Eidsvik \etal \cite{eidsvik}).  Another notable effort based on a modified model is the work of Lindgren \etal \cite{lindgren}, which gives a way of approximating Gaussian processes by Gaussian Markov random fields for specific kernels in the Mat\'ern family.  In practice, these methods perform quite well, but in our approach, we consider maximizing the true likelihood as opposed to alternative models or estimators.

\subsection{Outline}
The remainder of the paper is organized as follows.  In \cref{sec:factor} we review hierarchical matrix structure and outline the recursive skeletonization factorization \cite[section 3]{hifie}, which is our hierarchical matrix format of choice for fast MLE.  In \cref{sec:peel}, we discuss a modification of the matrix peeling algorithm by Lin \etal \cite{Lin2011}, which we use to {quickly evaluate} the gradient of the log-likelihood.  In \cref{sec:summary}, we succinctly summarize our framework.  In \cref{sec:results} we present numerical results on a number of test problems and demonstrate the scaling of our approach.  Finally, in \cref{sec:conclusion}, we make some concluding remarks.

\section{Factorization of the covariance matrix}\label{sec:factor}

{Consider the kernelized covariance matrix $\m{\Sigma}$ as in \eqref{eq:normaldist}, and assume for simplicity of exposition that the points $\{x_i\}_{i=1}^n$ are uniformly distributed inside a rectangular domain $\Omega$.  Partitioning $\Omega$ into four equal rectangular subdomains $\Omega_{1;i}$ for $i=1,\dots,4$, it has been observed that the corresponding block partitioning of $\m{\Sigma}$ exposes low-rank structure of off-diagonal blocks when $K(x,y;\theta)$ is sufficiently nice as a function of $\|x-y\|$.}

Concretely, we define the set $[n]\equiv\{1,2,\dots,n\}$ and let $\I_{1;i}\subset[n]$ denote the index set indexing degrees of freedoms (DOFs) located inside $\Omega_{1;i}$ for $i=1,\dots,4$ such that $\{x_j\}_{j\in\I_{1;i}} \subset\Omega_{1;i}$ and $\bigcup_{i=1}^4 \I_{1;i} = [n].$  In a self-similar fashion, we further partition $\Omega_{1;1}$ into the four subdomains $\Omega_{2;i}$ (with corresponding DOFs $\I_{2;i}$) for $i=1,\dots,4$ and obtain the decomposition shown in \cref{fig:rankstructure} (left).  {Assuming that the covariance kernel is smooth away from $x=y$ and does not exhibit high-frequency oscillations, the off-diagonal blocks $\m{\Sigma}(\I_{1;i},\I_{1;j})$ for $i\ne j$ and $\m{\Sigma}(\I_{2;i},\I_{2;j})$ for $i\ne j$ in the corresponding partitioning shown in the same figure (right) tend to be \emph{numerically low-rank} and thus compressible.}

\begin{definition}[Numerically low-rank]\label{def:lowrank}
We call a matrix $\m{A}\in\R^{m_1\times m_2}$ \emph{numerically low-rank} with respect to a specified tolerance $\epsilon$ if for some $r<\min(m_1,m_2)$ there exist matrices $\m{V_1}\in\R^{m_1\times r}$ and $\m{V_2}\in\R^{m_2\times r}$ such that $\|\m{A} - \m{V}_1\m{V}_2^T\|_2 \le \epsilon \|A\|_2{.}$
\end{definition}

{The} rank structure of \cref{fig:rankstructure} includes numerically low-rank blocks at multiple spatial scales \emph{independent of the length-scale of the underlying Gaussian process}, \ie, we may continue to recursively subdivide the domain and expose more compressible blocks of $\m{\Sigma}$.  Explicitly representing each of these blocks in low-rank form leads to the so-called hierarchical off-diagonal low-rank (HODLR) matrix format that has been used by Ambikasaran \etal \cite{ambikasaran} to compress various families of covariance kernels with application to Gaussian processes.

{
	\begin{remark}
	A sufficient condition to ensure this rank structure is that for any pair of distinct subdomains on the same level $\Omega_{\ell;i}$ and $\Omega_{\ell;j}$ there exists an approximation $K(x_i,x_j;\theta)\approx \sum_{k=1}^r f_k(x_i) g_k(x_j)$ for any $x_i\in\I_{\ell;i}$ and $x_j\in\I_{\ell;j}$ where the number of terms $r$ is relatively small (via, e.g., Chebyshev polynomials).  An important example where this is not typically the case is periodic kernels with short period relative to the size of the domain.  Further, the techniques we discuss here are less relevant to compactly-supported kernels, which are typically already efficient to compute with using standard sparse linear algebra.
	\end{remark}
}

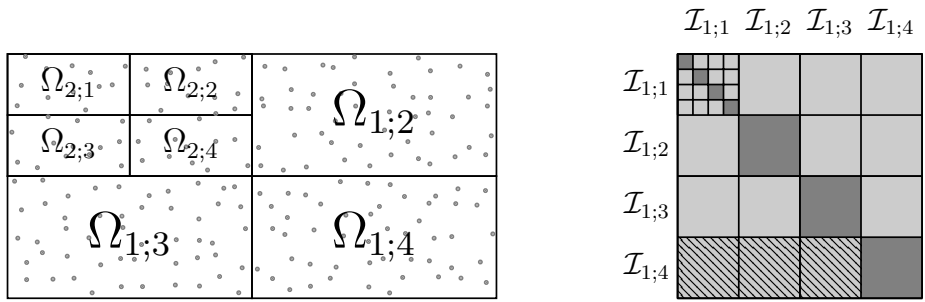
\begin{figure}[h]%
\centering
\scalebox{3.25}{
		\begin{tikzpicture}
			\draw[draw=none, use as bounding box](0,0) rectangle (2,1);
			\draw[very thin] (0,0) rectangle (1,0.5);
			\draw[very thin] (0,0.5) rectangle (1,1);
			\draw[very thin] (1,0) rectangle (2,0.5);
			\draw[very thin] (1,0.5) rectangle (2,1);

			\draw[very thin] (0,0.5) rectangle (0.5,0.75);
			\draw[very thin] (0,0.75) rectangle (0.5,1);
			\draw[very thin] (0.5,0.5) rectangle (1,0.75);
			\draw[very thin] (0.5,0.75) rectangle (1,1);

			\foreach\i in {0,...,19}{
				\foreach\j in {0,...,9}{
	    			\filldraw[color=darkgray] (1/64+\i/10 + 0.075*random,1/64+\j/10 + 0.075*random) circle[radius=3e-3]{};
	    		}
	    	}
	    	\draw (1.5,0.25) node[scale=0.5] {$\Omega_{1;4}$};
	    	\draw (1.5,0.75) node[scale=0.5] {$\Omega_{1;2}$};
	    	\draw (0.5,0.25) node[scale=0.5] {$\Omega_{1;3}$};

	    	\draw (0.75,0.625) node[scale=0.33] {$\Omega_{2;4}$};
	    	\draw (0.25,0.625) node[scale=0.33] {$\Omega_{2;3}$};
			\draw (0.75,0.875) node[scale=0.33] {$\Omega_{2;2}$};
			\draw (0.25,0.875) node[scale=0.33] {$\Omega_{2;1}$};
		\end{tikzpicture}
		}
\hspace{3em}
 	\scalebox{3.25}{
			\begin{tikzpicture}
			\draw[draw=none, use as bounding box](0,0) rectangle (1,1.125);
			\draw[draw=black, use as bounding box,very thin](0,0) rectangle (1,1);
			\fill[lightgray] (0,0) rectangle (1,1);
			\fill[darkgray] (0.25,0.5) rectangle (0.5,0.75);
			\fill[darkgray] (0.5,0.25) rectangle (0.75,0.5);
			\fill[darkgray] (0.75,0) rectangle (1,0.25);

			\fill[darkgray] (0,0.9375) rectangle (0.0625,1);
			\fill[darkgray] (0.0625,0.875) rectangle (0.125,0.9375);
			\fill[darkgray] (0.125,0.8125) rectangle (0.1875,0.875);
			\fill[darkgray] (0.1875,0.75) rectangle (0.25,0.8125);

			\draw[draw=none,very thin,pattern=north west lines, pattern color = black] (0,0) rectangle (0.75,0.25);

			\draw[step=0.5cm,black,very thin] (0,0) grid (1,1);
			\draw[step=0.25cm,black,very thin] (0,0.5) grid (0.5,1);
			\draw[step=0.25cm,black,very thin] (0.5,0) grid (1,0.5);

			\draw[step=0.25cm,black,very thin] (0.5,0.5) grid (1,1);
			\draw[step=0.25cm,black,very thin] (0,0) grid (0.5,0.5);

			\draw[step=0.0625cm,black,very thin] (0,0.75) grid (0.25,1);

			\foreach\j in {1,...,4}{
				\draw (-0.125 + \j * 0.25,1.125) node[scale=0.3] {$\I_{1;\j}$};
			}
			\foreach\j in {1,...,4}{
				\draw (-0.125,1.125-\j * 0.25) node[scale=0.3] {$\I_{1;\j}$};
			}

		\end{tikzpicture}
	}
\caption{To expose low-rank structure in the matrix $\m{\Sigma}$, consider the partitions $\Omega = \cup_{i=1}^4 \Omega_{1;i}$ and $\Omega_{1;1} = \cup_{i=1}^4\Omega_{2;i}$ (left), with corresponding index sets $\I_{\ell;i}$ for $\ell=1,\dots,2$, $i=1,\dots,4$.  Due to smoothness properties of $K(\cdot,\cdot;\theta)$, the blocks of $\m{\Sigma}$ are rank-structured (right), where dark gray blocks on the diagonal are full-rank and light gray off-diagonal blocks are numerically low-rank.  Representing each off-diagonal block in low-rank form independently leads to the HODLR format.  In contrast, the HBS format uses {one low-rank representation for all off-diagonal blocks for each block row}, \eg, the patterned blocks in the bottom row are aggregated into a single low-rank matrix.   \label{fig:rankstructure}}
\end{figure}

The HODLR format is only one of many hierarchical matrix formats, appearing as a special case of the $\H$- and $\H^2$-matrices of Hackbusch and collaborators \cite{Hackbusch,HackbuschK,HackbuschB}.  This format is particularly simple as at each level it compresses \emph{all off-diagonal blocks} of $\m{\Sigma}$, including those corresponding to domains that are adjacent (\eg, $\Omega_{1;2}$ and $\Omega_{1;3}$).  Matrices compressible in this way are referred to as \emph{weakly-admissible} \cite{Hackbuschweak}, in contrast to \emph{strongly-admissible} matrices which compress only a subset of off-diagonal blocks at each level.  Closely related literature includes work on hierarchically semiseparable (HSS) matrices \cite{xia,fastulv,fasthss} and hierarchically block separable (HBS) matrices \cite{martinsson-rokhlin,domainsAd} which offer simplified representations for weakly-admissible matrices with improved runtime.

For matrices where the entries are explicitly generated by an underlying kernel function such as $\m{\Sigma}$ and its derivatives, specific factorization algorithms have been developed to exploit this additional structure for increased efficiency \cite{gg,domainsAd,rskel,hifie,corona2013}.  These ``skeletonization-based'' algorithms, based on the framework introduced by Martinsson \& Rokhlin \cite{martinsson-rokhlin} stemming from observations by Starr \& Rokhlin \cite{starr_rokhlin} and Greengard \& Rokhlin \cite{greengard-rokhlin}, construct low-rank representations of certain off-diagonal blocks using the skeletonization
process described by Cheng \etal \cite{id}.

Our framework is agnostic to the choice of hierarchical factorization used for $\m{\Sigma}$ and its derivatives, provided that the factorization admits fast linear algebra computations (including computation of the log-determinant) with the underlying operator.  The {recursive skeletonization factorization} that we use in this paper was first introduced by {Ho \& Ying \cite[section 3]{hifie}} as a multiplicative factorization based on skeletonization \cite{martinsson-rokhlin}.  {In the remainder of this section we provide a brief review of the algorithm.}

\begin{remark}\label{rem:tree}
In what follows, we will continue to assume that the DOFs $\{x_i\}_{i=1}^n$ are uniformly distributed inside a rectangular subdomain $\Omega$ for simplicity of exposition.  Further, we we will describe the algorithm as though the quadtree representing the hierarchical partitioning of space is \emph{perfect}, \ie, every subdomain is subdivided into four child subdomains at every level.  {I}n practice an adaptive decomposition of space is used {to avoid subdividing} the domain in regions {of low observation density}.
\end{remark}

\subsection{Block compression through skeletonization}\label{sec:blockcompression}
We begin by recursively subdividing $\Omega$ into four subdomains until each leaf-level subdomain contains a constant number of observations independent of $n$.  This leads to a quadtree data structure with levels labeled $\ell=0$ through $\ell=L$, where $\ell=0$ refers to the entire domain $\Omega$ and $\ell=L$ refers to the collection of subdomains $\Omega_{L;i}$ for $i=1,\dots,4^{L}$.  Considering factorization of the covariance matrix $\m{\Sigma}$, the basic intuition of the method is to first compress all blocks of $\m{\Sigma}$ corresponding to covariances between observations in distinct subdomains at the leaf level, and then to recurse in a bottom-up traversal{.}

Consider first a single leaf-level subdomain containing observations indexed by $\I\subset[n]$ and define the complement DOF set $\I^c\equiv[n]\setminus \I$.  Given a specified tolerance $\epsilon>0$, the algorithm proceeds by compressing the off-diagonal blocks $\m{\Sigma}(\I,\I^c)$ and $\m{\Sigma}(\I^c,\I)$ {as in the HBS format (see \cref{fig:rankstructure})} through the use of an \emph{interpolative decomposition} (ID) \cite{id}.
\begin{definition}[Interpolative decomposition]\label{def:id}
  Given a matrix $\m{A}\in\R^{m\times |\I|}$ with columns indexed by $\I$ and a tolerance $\epsilon > 0$, an $\epsilon$-accurate interpolative decomposition of $\m{A}$ is a partitioning of $\I$ into DOF sets
  associated with so-called skeleton columns $\sk\subset \I$ and redundant columns $\rd = \I \setminus \sk$ and a corresponding interpolation matrix $\,\m{T}_\I$ such that $\|\m{A}(:,\rd)- \m{A}(:,\sk)\m{T}_\I\|_2 \le \epsilon\|\m{A}\|_2,$
where $\m{A}(:,\rd)\in\R^{m\times|\rd|}$ is given by subselecting the columns of $\m{A}$ indexed by $\rd$, and $\m{A}(:,\sk)$ is defined analogously.
\end{definition}
{It is desirable in \cref{def:id}} to take $|\sk|$ as small as possible for a given $\epsilon$.

Given an ID of $\m{\Sigma}(\I^c,\I)$ such that $\m{\Sigma}(\I^c,\rd) \approx \m{\Sigma}(\I^c,\sk)\m{T}_\I$, $\m{\Sigma}$ can be written {in block form (up to a permutation)} as
\begin{align*}
\left[ \begin{array}{rl|l}
\m{\Sigma}(\I^c,\I^c) & \m{\Sigma}(\I^c,\sk) & \m{\Sigma}(\I^c,\rd) \\
\m{\Sigma}(\sk,\I^c) & \m{\Sigma}(\sk,\sk) & \m{\Sigma}(\sk,\rd) \\ \hline
\m{\Sigma}(\rd,\I^c) & \m{\Sigma}(\rd,\sk) & \m{\Sigma}(\rd,\rd)
\end{array}\right]\approx\left[ \begin{array}{rl|l}
\m{\Sigma}(\I^c,\I^c) & \m{\Sigma}(\I^c,\sk) & \m{\Sigma}(\I^c,\sk)\m{T}_\I \\
\m{\Sigma}(\sk,\I^c) & \m{\Sigma}(\sk,\sk) & \m{\Sigma}(\sk,\rd) \\ \hline
\m{T}_\I^T\m{\Sigma}(\sk,\I^c) & \m{\Sigma}(\rd,\sk) & \m{\Sigma}(\rd,\rd)
\end{array}\right].
\end{align*}
Using a sequence of block row and column operations, we first eliminate the blocks $\m{\Sigma}(\I^c,\sk)\m{T}_\I$ and $\m{T}_\I^T\m{\Sigma}(\sk,\I^c)$ and then decouple the bottom-right {block to} obtain
\begin{align*}
\m{L}_\I^{-1}\m{\Sigma}\m{L}_\I^{-T} &\approx
\left[ \setlength\arraycolsep{4pt} \begin{array}{rc|l}
\m{\Sigma}(\I^c,\I^c) & \m{\Sigma}(\I^c,\sk) &\\
\m{\Sigma}(\sk,\I^c) & \m{\Sigma}(\sk,\sk) &\m{X}_{\sk\rd}  \\ \hline
&  \m{X}_{\rd\sk}& \m{X}_{\rd\rd}
\end{array}\right] 
= \m{U}_\I
\left[\setlength\arraycolsep{4pt}\begin{array}{rc|l}
\m{\Sigma}(\I^c,\I^c) & \m{\Sigma}(\I^c,\sk) &\\
\m{\Sigma}(\sk,\I^c) & \m{X}_{\sk\sk} &  \\ \hline
&  & \m{X}_{\rd\rd}
\end{array}\right]\m{U}_\I^T,
\end{align*}
where $\m{L}_\I$ and $\m{U}_\I$ are block unit-triangular matrices that are fast to apply or invert and the $\m{X}$ subblocks are linear combinations of the $\m{\Sigma}$ subblocks.

\subsection{The recursive skeletonization factorization}
Defining the collection of DOF sets corresponding to subdomains at level $\ell=L$ as $\L_L\equiv\{\I_{L;1},\I_{L;2},\dots,\I_{L;4^L}\}$, we use the skeletonization process of \cref{sec:blockcompression} to compress the corresponding blocks of $\m{\Sigma}$ for each $\I\in\L_L$, yielding $\m{\Sigma} \approx \left(\prod_{\I\in\L_L} \m{L}_\I \m{U}_\I\right) \tilde{\m{\Sigma}}_L \left(\prod_{\I\in\L_L} \m{U}_\I^T \m{L}_\I^T\right),$
where the order taken in the product over $\L_L$ does not matter due to the structure of the $\m{L}$ and $\m{U}$ matrices.
Using $\rd_{L;i}$ and $\sk_{L;i}$ to denote the redundant DOFs and skeleton DOFs associated with $\I_{L;i}$ for each $i=1,\dots,4^L$ and defining $\rd_L\equiv \cup_{i=1}^{4^L} \rd_{L;i}$, the {blocks of the} matrix $\tilde{\m{\Sigma}}_L$ {have} the following structure {for each $i=1,\dots,4^L$}:
\begin{itemize}
\item The {modified block} $\tilde{\m{\Sigma}}_L(\rd_{L;i},\rd_{L;i})$ {has been decoupled} from the rest of $\tilde{\m{\Sigma}}_L$.
\item The block $\tilde{\m{\Sigma}}_L(\sk_{L;i},\sk_{L;i})$ has been modified.
\item The blocks $\tilde{\m{\Sigma}}_L(\sk_{L;i},\I^c_{L;i}\setminus\rd_L) = \m{\Sigma}(\sk_{L;i},\I^c_{L;i}\setminus\rd_L)$ and $\tilde{\m{\Sigma}}_L(\I_{L;i}^c\setminus\rd_L,\sk_{L;i}) = \m{\Sigma}(\I^c_{L;i}\setminus\rd_L,\sk_{L;i})$ remain unmodified from what they were in $\m{\Sigma}$.
\end{itemize}
In other words, we have identified and decoupled all redundant DOFs at the leaf level while leaving unchanged the blocks of $\m{\Sigma}$ corresponding to kernel evaluations between skeleton DOFs in distinct leaf-level subdomains.

\begin{figure}%
\centering
\begin{minipage}{0.3\textwidth}%
\centering
\includegraphics[width=\textwidth]{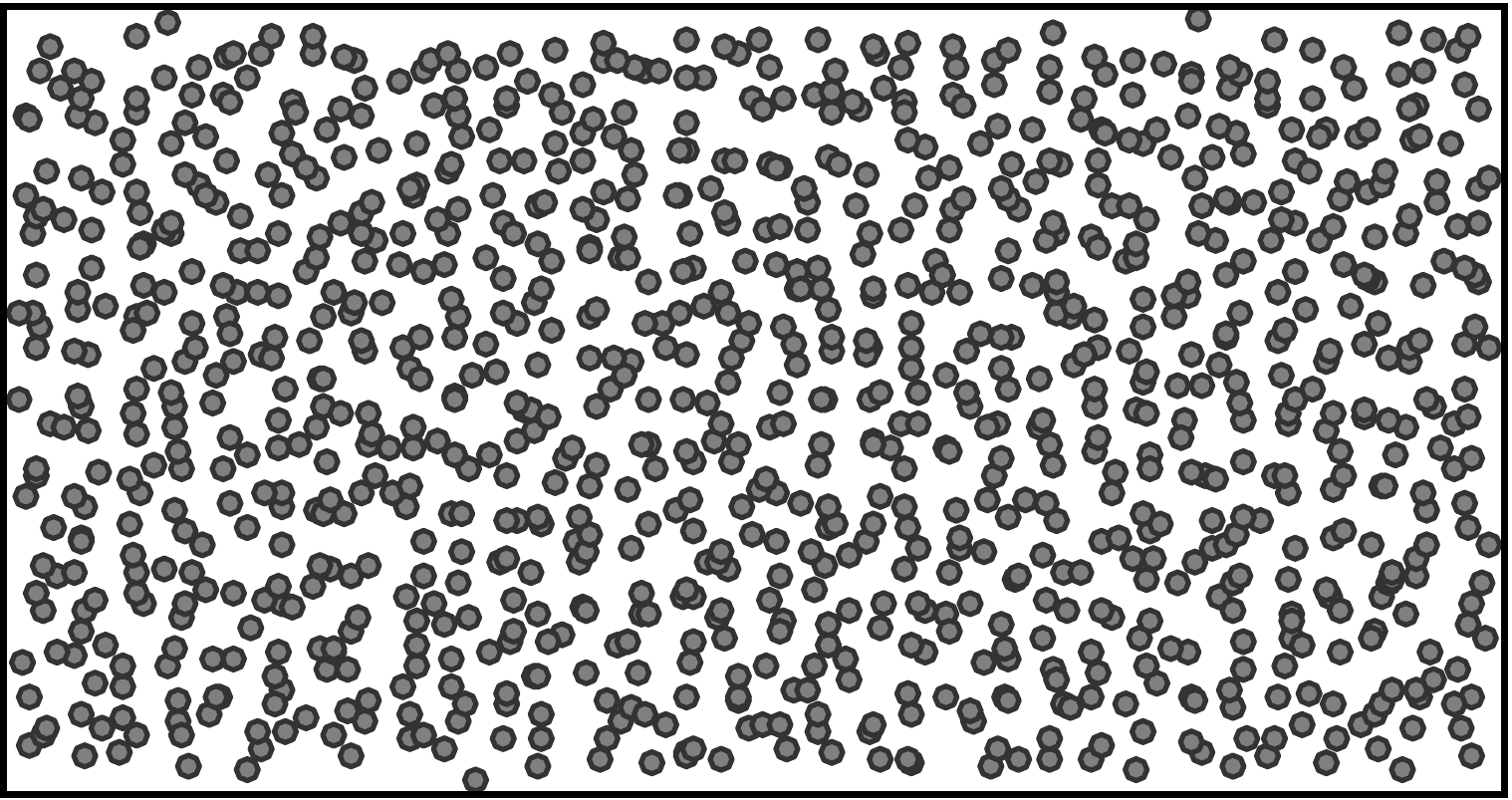}
\end{minipage}%
\quad
\begin{minipage}{0.3\textwidth}%
\centering
\includegraphics[width=\textwidth]{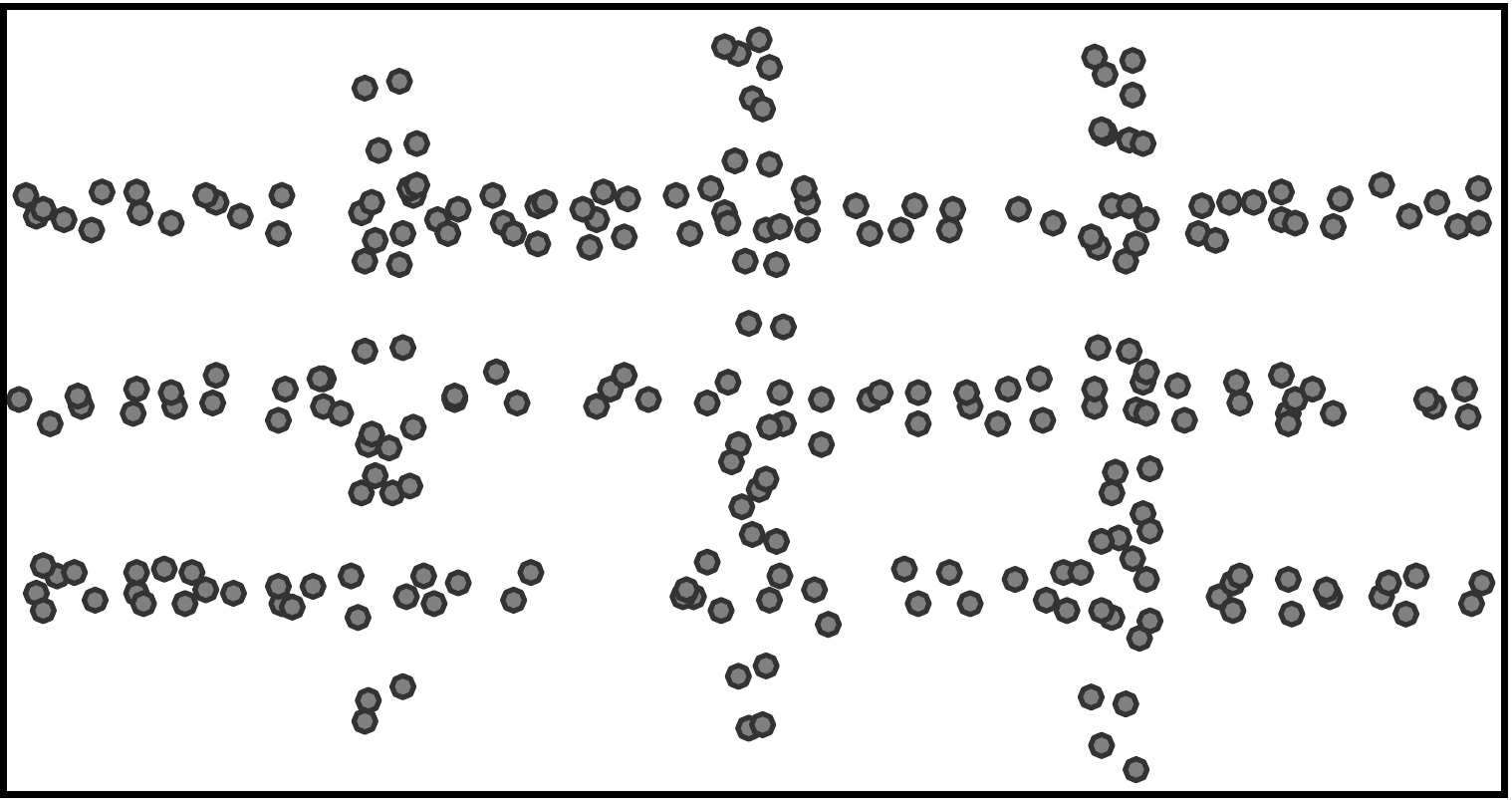}
\end{minipage}%
\quad
\begin{minipage}{0.3\textwidth}%
\centering
\includegraphics[width=\textwidth]{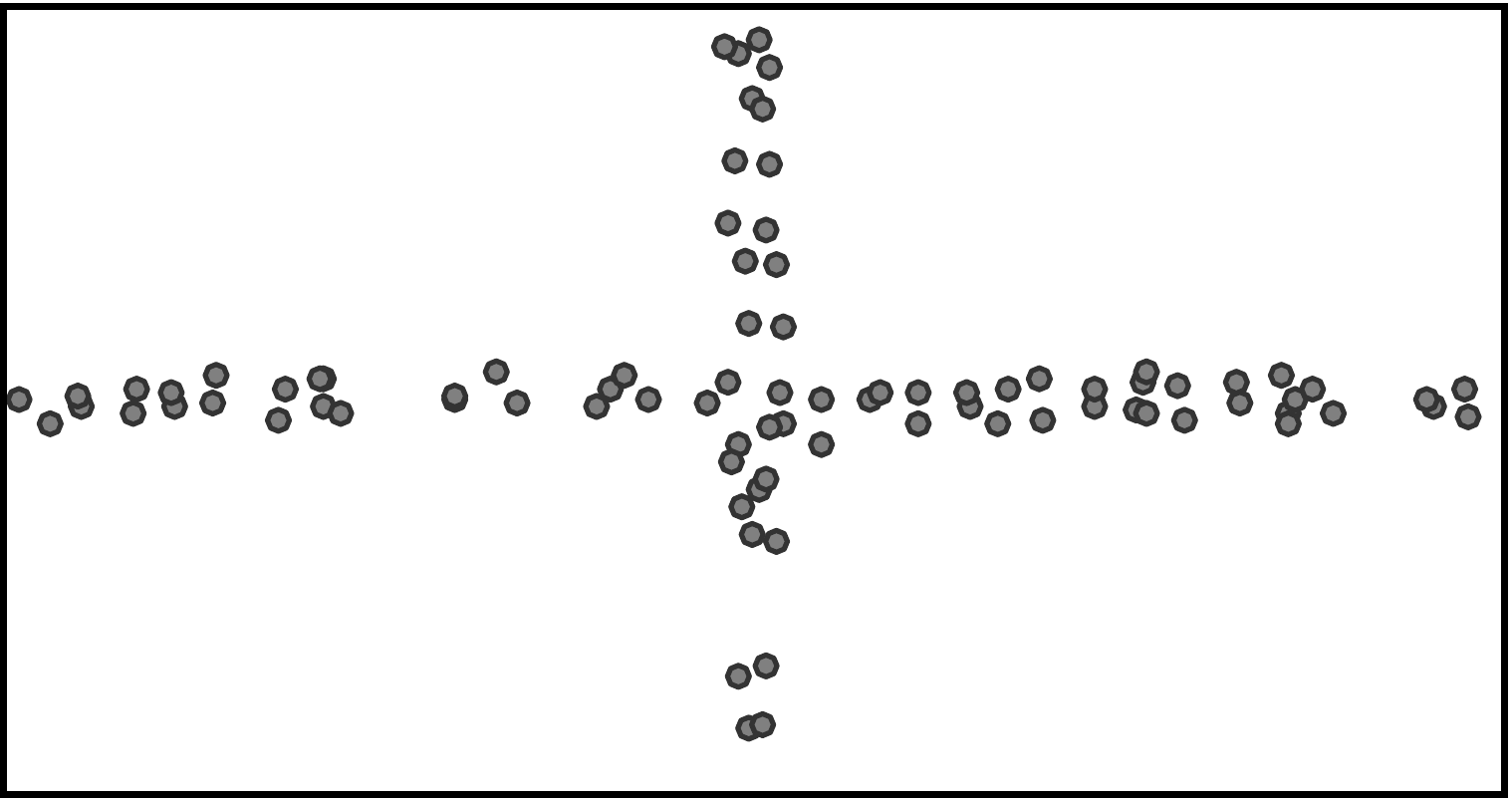}
\end{minipage}%
\caption{Shown here are the DOFs to be compressed at each level of the recursive skeletonization factorization.  At the leaf level (left), all DOFs are involved in skeletonization.  At each subsequent level (center, right), only skeleton DOFs from the previous level are involved in further skeletonization.  We see that the skeleton DOFs tend to line the boundaries of their corresponding subdomains. \label{fig:rskelfdofs}}
\end{figure}

The recursive skeletonization factorization of $\m{\Sigma}$ is given by repeating this process at each higher level of the quadtree.  For example, at level $\ell=L-1$ of the quadtree, each subdomain contains skeleton DOFs corresponding to its four distinct child subdomains in the tree.  However, the redundant DOFs of its child subdomains no longer need to be considered as they have already been decoupled.  We thus define $\tilde \I_{L-1;i}\equiv~\I_{L-1;i}\setminus\rd_L$ for each $i=1,\dots,4^{(L-1)}$ and write the collection of DOFs remaining at this level as $\L_{L-1}\equiv\{\tilde \I_{L-1;1}, \tilde \I_{L-1;2},\dots,\tilde\I_{L-1;4^{(L-1)}}\}$.  Due to the hierarchical block low-rank structure of $\m{\Sigma}$, off-diagonal blocks at this level are again compressible through skeletonization, yielding
\begin{align*}
\tilde{\m{\Sigma}}_L &\approx \m{P}_L\left(\prod_{\I\in\L_{L-1}} \m{L}_\I \m{U}_\I\right)\tilde{\m{\Sigma}}_{L-1} \left(\prod_{\I\in\L_{L-1}} \m{U}_\I^T \m{L}_\I^T\right)\m{P}_L^T,
\end{align*}
where $\m{P}_L$ is a global permutation matrix regrouping the DOF sets in $\L_{L-1}$ to be contiguous.  Proceeding in this fashion level-by-level {and defining the ordered product $\prod_{\ell=1}^L \m{A}_\ell \equiv  \m{A}_L\m{A}_{L-1}\dots\m{A}_{1}$,} we obtain the full recursive skeletonization factorization $\m{F}$ of $\m{\Sigma}$
\begin{align}\label{eq:rskelf}
\m{\Sigma} &\approx \left[\prod_{\ell=1}^{L} \left(\prod_{\I\in\L_\ell} \m{L}_\I \m{U}_\I\right)\m{P}_\ell\right] \tilde{\m{\Sigma}}_{1}\left[\prod_{\ell=1}^{L} \left(\prod_{\I\in\L_\ell} \m{L}_\I \m{U}_\I\right)\m{P}_\ell\right]^T\notag \\
&= \left[\prod_{\ell=1}^{L} \left(\prod_{\I\in\L_\ell} \m{L}_\I \m{U}_\I\right)\m{P}_\ell\right] \m{C}\m{C}^T\left[\prod_{\ell=1}^{L} \left(\prod_{\I\in\L_\ell} \m{L}_\I \m{U}_\I\right)\m{P}_\ell\right]^T \equiv \m{F},
\end{align}
where $\tilde{\m{\Sigma}}_{1}$ is block-diagonal with diagonal blocks corresponding to the sets of redundant DOFs at each level and $\tilde{\m{\Sigma}}_{1}=\m{C}\m{C}^T$ is the Cholesky decomposition of $\tilde{\m{\Sigma}}_1$.

\begin{remark}
Because the factorization $\m{F}\approx\m{\Sigma}$ is approximate, using an extremely inaccurate tolerance $\epsilon$ in the IDs of \cref{def:id} admits the possibility that $\tilde{\m{\Sigma}}_1$ may be slightly indefinite due to approximation error.  In practice, this is not an issue for any tolerance precise enough to be used for computing $\ell(\theta)$ for optimization purposes.  When factoring the derivative matrices $\m{\Sigma}_i$ for $i=1,\dots,p$ (which may themselves be indefinite), the Cholesky decomposition may be replaced with, \eg, an LDL$^T$ factorization.
\end{remark}

\subsection{Computational complexity}\label{sec:rskelfscaling}
The computational cost of the recursive skeletonization factorization is in theory dominated by the cost of computing IDs $\m{\Sigma}(\I^c,\rd)\approx\m{\Sigma}(\I^c,\sk)\m{T}_\I$ of $\m{\Sigma}(\I^c,\I)$ in \cref{sec:blockcompression} for each $\I$.  This is because the typical algorithm to compute an ID is based on a rank-revealing QR factorization, such that the $m \times |\I|$ ID of \cref{def:id} has complexity $\O(m|\I|^2)$ \cite{id}.  This dependence on $m$ is prohibitively expensive during initial levels of the factorization, since $m=|\I^c| = \O(n)$ and there are $\O(n)$ such IDs to compute.

The original application of skeletonization was to boundary integral equations arising from elliptic partial differential equations, in which case the so-called ``proxy trick'' described by Martinsson \& Rokhlin \cite{martinsson-rokhlin} can be applied to accelerate the computation of an ID through the use of integral identities.  These integral identities do not strictly apply in the case where $K(\cdot,\cdot;\theta)$ is a general covariance function, but we find that a variant of this proxy trick works well to obtain similar acceleration.

\subsubsection{Modified proxy trick}\label{sec:proxy}
Suppose that the index set $\I$ corresponds to points inside the subdomain $\rm{B}\subset\Omega$ in \cref{fig:proxy}, where we use $\rm{B}$ as a stand-in for an arbitrary subdomain $\Omega_{\ell;i}$ in our quadtree.  The purpose of computing an ID of $\m{\Sigma}(\I^c,\I)$ is to find a small set of skeleton DOFs $\sk\subset\I$ such that the range of $\m{\Sigma}(\I^c,\sk)$ approximately captures the range of $\m{\Sigma}(\I^c,\I)$.  The key to computational acceleration using the proxy trick is to accomplish this without using all rows of $\m{\Sigma}(\I^c,\I)$ in the computation.

As detailed by Ho \& Ying \cite[section 3.3]{hifie}, we can partition the DOFs $\I^c$ into those that are near to $\rm{B}$ and those that are far from $\rm{B}$, denoted $\N$ and $\F$ respectively.  For example, we may take $\N$ to be all points $x_i\in\I^c$ such that the distance between $x_i$ and the center of $\rm{B}$ is less than some radius.  The proxy trick proceeds by finding a surrogate representation $\m{M}(\Gamma,\I)$ for $\m{\Sigma}(\F,\I)$ in the ID computation, such that $\m{M}(\Gamma,\I)$ has many fewer rows than $\m{\Sigma}(\F,\I)$ and
\begin{align}\label{eq:proxyid}
\left[\begin{array}{c}\m{\Sigma}(\N,\rd)\\ \m{M}(\Gamma,\rd) \end{array}\right] \approx \left[\begin{array}{c}\m{\Sigma}(\N,\sk)\\ \m{M}(\Gamma,\sk) \end{array}\right]\m{T}_\I \implies \left[\begin{array}{c}\m{\Sigma}(\N,\rd)\\ \m{\Sigma}(\F,\rd) \end{array}\right] \approx \left[\begin{array}{c}\m{\Sigma}(\N,\sk)\\ \m{\Sigma}(\F,\sk) \end{array}\right]\m{T}_\I 
\end{align}
such that we may compute the left ID in \eqref{eq:proxyid} and get the right ID for ``free''.  

In the modified proxy trick, we let $\Gamma$ be a set of $n_\text{prox}$ points discretizing the gray annulus in the right of \cref{fig:proxy}.  {Crucially}, this differs from the original proxy trick due to the fact that we discretize a two-dimensional region (the annulus), whereas if our kernel satisfied some form of a Green's identity we could instead discretize a quasi-one-dimensional curve (a circle) around B as in the original proxy trick. Defining the matrix $\m{M}(\Gamma,\I)$ to have entries $K(y,x_i;\theta)$ with rows indexed by $y\in\Gamma$ and columns indexed by $x_i\in\I$, we observe that \eqref{eq:proxyid} holds without significant loss in accuracy even with $n_\text{prox}$ relatively small.  This brings the complexity of computing the right ID down to $\O(|\I|^3 + n_\text{prox}|\I|^2)$, which is beneficial when $n_\text{prox}$ is small compared to $|\I^c|$.  In practice, we take $n_\text{prox}$ to be constant with respect to the total number of points $n$.

\begin{figure}[ht]
\centering
\scalebox{2.75}{
		\begin{tikzpicture}
		\pgfmathsetseed{1234}
			\draw[draw=none, use as bounding box](0,0) rectangle (2,1);
			\draw[very thin] (0,0) rectangle (2,1);

			\foreach\i in {0,...,19}{
				\foreach\j in {0,...,9}{
	    			\filldraw[color=darkgray] (1/64+\i/10 + 0.075*random,1/64+\j/10 + 0.075*random) circle[radius=5e-3]{};
	    		}
	    	}
	    	\filldraw[color=white] (1-0.2,0.5-0.2) rectangle (1+0.2,0.5+0.2);
	    	\draw[very thin] (1-0.2,0.5-0.2) rectangle (1+0.2,0.5+0.2);
	    	\draw (1.0,0.5) node[scale=0.5] {B};
		\end{tikzpicture}
		}
		\hspace{1mm}
		\scalebox{2.75}{
		\begin{tikzpicture}
		\pgfmathsetseed{1234}
			\draw[draw=none, use as bounding box](0,0) rectangle (2,1);
			\draw[very thin] (0,0) rectangle (2,1);

			\filldraw[color=lightgray] (1,0.5) circle (0.49);
			\filldraw[color=white] (1,0.5) circle (0.35);

			\begin{scope}
			\clip (1,0.5) circle (0.36);
			\foreach\i in {0,...,19}{
				\foreach\j in {0,...,9}{
	    			\filldraw[color=darkgray] (1/64+\i/10 + 0.075*random,1/64+\j/10 + 0.075*random) circle[radius=5e-3]{};
	    		}
	    	}
	    	\end{scope}

	    	\draw[very thin] (1,0.5) circle (0.36);

	    		    	\filldraw[color=white] (1-0.2,0.5-0.2) rectangle (1+0.2,0.5+0.2);
	    	\draw[very thin] (1-0.2,0.5-0.2) rectangle (1+0.2,0.5+0.2);
	    	\draw (1.0,0.5) node[scale=0.5] {B};

	    	\draw[very thin] (1,0.5) circle (0.49);

	    	\foreach\r in {0.375,0.425,...,0.475}{
				\foreach\theta in {0,15,...,360}{
				\filldraw[color=black,shift={(1,0.5)}](\theta:\r) circle[radius=5e-3] {};
				}
				}

		\end{tikzpicture}
		}
		\caption{Because of the underlying kernel, computing an ID of the submatrix $\m{\Sigma}(\I^c,\I)$ can be accelerated, where $\I\subset[n]$ indexes observations inside the subdomain $\textrm{B}$ (left).  Rather than considering all of $\I^c$, the (modified) \emph{proxy trick} involves neglecting rows of  $\m{\Sigma}(\I^c,\I)$ corresponding to points $\F\subset \I^c$ far from $\textrm{B}$.  Instead, the ID is computed by considering only points $\N\subset\I^c$ near $\textrm{B}$ combined with a small number of so-called ``proxy points'' discretizing an annulus around $\textrm{B}$ (right).\label{fig:proxy}}
\end{figure}
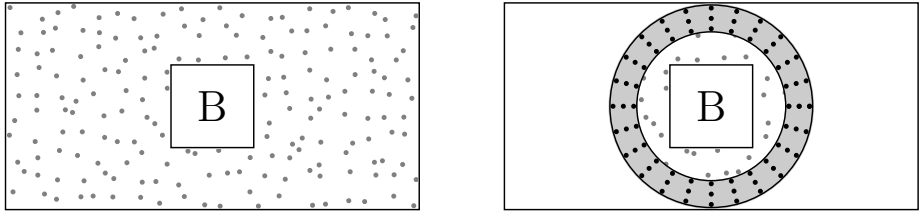

\subsubsection{Complexity sketch using the modified proxy trick}
Using the modified proxy trick, the cost of the recursive skeletonization factorization is essentially determined by the number of DOFs interior to each skeletonized subdomain, \ie, $|\sk_{\ell;i}|$ for each $\ell=1,\dots,L$ and $i=1,\dots,4^\ell$. As seen in \cref{fig:rskelfdofs}, the skeleton DOFs tend to line the boundaries of their corresponding subdomains.  This is statistically intuitive: due to the fact that our kernels of interest $K(\cdot,\cdot;\theta)$ decay smoothly, the subset of DOFs that best represent the covariance structure of a subdomain with the rest of the domain is the subset closest to the rest of the domain.

Assuming a uniform distribution of points $\{x_i\}_{i=1}^n$ and perfect quadtree, the average number of skeleton DOFs per subdomain at level $\ell$ is on the order of the sidelength of a subdomain in the quadtree at level $\ell$.  In other words, the number of skeleton DOFs per box grows by roughly a factor of two each time we step up a level in the tree and thus $s_\ell \equiv \frac{1}{4^\ell}\sum_{i=1}^{4^\ell} |\sk_{\ell;i}| = \O(2^{-\ell}).$
The assumptions that lead to this rank growth bound are described in more detail by Ho \& Greengard \cite[section 4]{rskel}; we do not go into them here.

\begin{theorem}[\cite{martinsson-rokhlin,hifie,rskel}]\label{thm:rs}
Assuming that the size of the skeleton sets behaves like $s_\ell=\O(2^{-\ell})$ for $\ell=1,\dots,L$ and $L\sim \log n$, the computational complexity of the recursive skeletonization factorization $\m{F}\approx\m{\Sigma}$ (with constants depending on the tolerance $\epsilon$ in \cref{def:id}) is $T_{\text{factor}} = \O(n^{3/2})$  and $T_{\text{apply}} = T_{\text{solve}} = \O(n\log n)$,
where $T_\text{factor}$ is the complexity of the factorization and $T_{\text{apply}}$ and $T_{\text{solve}}$ are the complexities of applying $\m{F}$ or $\m{F}^{-1}$ to a vector.  The storage complexity is $\O(n\log n)$.  
\end{theorem}

From \eqref{eq:rskelf} we see that the application of the factorization $\m{F}$ to a vector $x\in\R^n$ simply requires application of the block unit-triangular matrices $\m{L}_\I$ and $\m{U}_\I$ corresponding to each subdomain at each level as well as the block-diagonal Cholesky factor $\m{C}$ of $\tilde{\m{\Sigma}}_1$.  Further, the inverse of $\m{F}$ can be applied by noting that
\begin{align*}
\m{F}^{-1} &= \left[\prod_{\ell=L}^{1} \left(\prod_{\I\in\L_\ell} \m{P}_\ell^T\m{U}_\I^{-1}\m{L}_\I^{-1}\right)\right]^T \m{C}^{-T}\m{C}^{-1}\left[\prod_{\ell=L}^{1} \left(\prod_{\I\in\L_\ell} \m{P}_\ell^T\m{U}_\I^{-1}\m{L}_\I^{-1}\right)\right].
\end{align*}
Additionally, a generalized square root $\m{F}^{1/2}$ such that $\m{F}=\m{F}^{1/2}(\m{F}^{1/2})^T$ can be applied (as can its transpose or inverse) by taking
\begin{align*}
\m{F}^{1/2} &= \left[\prod_{\ell=1}^{L} \left(\prod_{\I\in\L_\ell} \m{L}_\I \m{U}_\I\right)\m{P}_\ell\right] \m{C}.
\end{align*}
Finally, the log-determinant of $\m{\Sigma}$ can be approximated by $\log |\m{\Sigma}| \approx \log |\m{F}| = 2\log |\m{C}|.$
\Cref{tab:complexity} summarizes the computational complexities for these operations, which essentially follow from \cref{thm:rs}.

\begin{table}
\caption{Complexity of the 2D recursive skeletonization factorization \label{tab:complexity}}
\centering
\ra{1.0}
\begin{tabular}{ll} \toprule
Operation               & Complexity \\ \midrule
Construct $\m{F} \approx \m{\Sigma}$ &     $\O\left(n^{3/2}\right) $     \\
Apply $\m{F}$ or $\m{F}^{-1}$ to a vector                          &  $\O(n\log n)$       \\ 
Apply $\m{F}^{1/2}$ or $\m{F}^{-1/2}$ to a vector                    & $\O(n \log n)$\\ 
Compute $\log |\m{F}|$ from $\m{F}$ & $\O(n\log n)$ \\ \bottomrule{}
\end{tabular}
\end{table}

By constructing the recursive skeletonization factorizations $\m{F}$ of $\m{\Sigma}$ and $\m{F}_i$ of $\m{\Sigma}_i$ for $i=1,\dots,p$, we see that after the $\O(n^{3/2})$ initial factorization cost each term in the evaluation of the log-likelihood or its gradient can be computed with cost $\O(n\log n)$ except for the product traces $\Tr(\m{\Sigma}^{-1}\m{\Sigma}_i)$ for $i=1,\dots,p$.  Further, through the approximate generalized square root $\m{F}^{1/2}$ of $\m{\Sigma}$ we can quickly sample from the distribution $N(0,\m{\Sigma})$.

\begin{remark}
While the recursive skeletonization factorization described here exploits the most well-justified rank assumptions on the covariance kernel $K(\cdot,\cdot;\theta)$, each recursive skeletonization factorization in our framework can be replaced by the closely-related \emph{hierarchical interpolative factorization} (HIF) \cite{hifie} or \emph{strong recursive skeletonization factorization} \cite{rss}, which are observed in practice to exhibit better scaling properties and also admit simple log-determinant computation.
\end{remark}

\section{Computing the trace terms}\label{sec:peel}
There are a number of methods for estimating the term $\Tr(\m{\Sigma}^{-1}\m{\Sigma}_i)$ appearing in each gradient component $g_i$ for $i=1,\dots,p$.  Employing the recursive skeletonization factorizations $\m{F}$ of $\m{\Sigma}$ and $\m{F}_i$ of $\m{\Sigma}_i$, the product $\m{\Sigma}^{-1}\m{\Sigma}_i\approx \m{F}^{-1}\m{F}_i$ (or a symmetrized form with the same trace) can be applied to a vector with complexity $\O(n\log n)$.  Using $\m{G}$ to denote this black-box linear operator, the classical statistical approach is the estimator of Hutchinson \cite{hutchinson}.  Drawing random vectors $u_i \in \{-1,1\}^n$ for $i=1,\dots,q$ such that the components of $u_i$ are independent and take value $\pm1$ with equal probability, the Hutchinson trace estimator is
\begin{align}\label{eq:hutch}
\Tr(\m{G}) \approx \frac{1}{q}\sum_{i=1}^q u_i^T\m{G}u_i,
\end{align}
which is unbiased with variance decaying as $1/q$ and has cost $\O(qn\log n)$.  For low-accuracy estimates of the trace, the Hutchinson estimator is simple and computationally efficient, but for higher accuracy it proves computationally infeasible to use the Hutchinson approach because of the slow rate of convergence in $q$, see \cref{sec:results}.

When $\m{\Sigma}$ and $\m{\Sigma}_i$ have hierarchical {rank structure}, it is reasonable to also look for hierarchical rank structure in the product $\m{\Sigma}^{-1}\m{\Sigma}_i$, as matrix inversion and multiplication preserve such rank structure (albeit with different ranks) in many cases \cite{hackbook}.  In our framework, we use the matrix peeling algorithm of Lin \etal \cite{Lin2011} for constructing an explicit $\H$-matrix representation of a fast black-box operator $\m{G}$.  At a high level, the method proceeds by applying the operator $\m{G}$ to random vectors drawn with a specific sparsity structure to construct an approximate representation of the off-diagonal blocks at each level.  We recursively perform low-rank compression level-by-level, following the same quadtree hierarchy as in the recursive skeletonization factorization, albeit in a top-down traversal rather than bottom-up.  Finally, at the bottom level of the tree, the diagonal blocks corresponding to leaf-level subdomains can be extracted and their traces computed.  While the full algorithm is applicable to both the strongly-admissible and weakly-admissible setting, the version of the algorithm we detail here is efficient for the simple weakly-admissible case.  We point the reader to Lin \etal \cite{Lin2011} for more details related to the modifications required for strong admissibility.

The use of a randomized method for computing low-rank representations of matrices, which we review below, is integral to the peeling algorithm.

\subsection{Randomized low-rank approximations}\label{sec:randomlra}
To begin, suppose that matrix $\m{A}\in\R^{m_1\times m_2}$ has (numerical) rank $r$ and that we wish to construct an explicit rank-$r$ approximation $\m{A}\approx \m{U}_1\m{M}\m{U}_2^T$
with $\m{U}_1\in\R^{m_1\times r}$, $\m{U}_2\in\R^{m_2\times r},$ and $\m{M}\in\R^{r\times r}$.  In the context of approximating the trace terms, for example, $\m{A}$ will be an off-diagonal block of $\m{\Sigma}^{-1}\m{\Sigma}_i$ or perhaps of a related symmetrized form with the same trace. Here we provide an overview of an algorithm that accomplishes this goal.

We begin by constructing approximations to the column space and row space of the matrix $\m{A}$. Let $c$ be a small integer and suppose $\m{W}_1 \in\R^{m_2\times(r+c)}$ and $\m{W}_2\in\R^{m_1\times(r+c)}$ are appropriately chosen random matrices, the distribution of which we will discuss later. Following Halko \etal \cite{halko}, let $\m{U}_1$ be a well-conditioned basis for the column space of $\m{A}\m{W}_1$ and $\m{U}_2$ be a well-conditioned basis for the column space of $\m{A}^T\m{W}_2$ constructed via, \eg, column-pivoted QR factorizations.  Using the Moore-Penrose pseudoinverse, we obtain a low-rank approximation according to the approach summarized by Lin \etal \cite[subsection 1.2]{Lin2011} via
\begin{equation}\label{eq:lowrankapprox}
\m{A} \approx \m{U}_1 \left[(\m{W}_2^T\m{U}_1)^\dagger (\m{W}_2^T\m{A}\m{W}_1)(\m{U}_2^T\m{W}_1)^\dagger\right] \m{U}_2^T = \m{U}_1 \m{M} \m{U}_2^T.
\end{equation}
Perhaps surprisingly, with an appropriate choice of $W_1$ and $W_2$ it is the case that with high probability this approximation is \emph{near-optimal}, in the sense that
\begin{align*}
\|\m{A}-\m{U}_1 \m{M} \m{U}_2^T\|_2 \le \alpha(m_1,m_2,c) \|\m{A} - \m{A}_{r,\text{best}}\|_2,
\end{align*}
where $\m{A}_{r,\text{best}}$ is the best rank-$r$ approximation of $\m{A}$ and $\alpha(m_1,m_2,c)$ is a small factor dependent on $c$ and the size of $\m{A}$.  Further, the approximation process can be monitored and controlled adaptively to ensure a target desired accuracy \cite{halko}.

It remains to discuss the choice of distribution for $\m{W}_1$ and $\m{W}_2$. The most common and straightforward choice is for both to have i.i.d.\ $N(0,1)$ entries, which guarantees the strongest analytical error bounds and highest success probability.  Under this choice, one can show that the algorithm as stated takes $\O(T_\text{apply}r + nr^2 + r^3)$, where $T_\text{apply}$ is the complexity of applying $\m{A}$ to a vector.  This is sufficiently fast for our purposes, though we note that it is possible to accelerate this using other distributions \cite{rokhlin-tygert,halko,tropp}.

\subsection{Matrix peeling for weakly-admissible matrices}\label{sec:peelsub}

For simplicity, we assume a perfect quadtree as in \cref{rem:tree}.  Further, we will assume that the numerical ranks of the off-diagonal blocks to a specified tolerance $\epsilon_\text{peel}$  are known $\emph{a priori}$ at each level, such that off-diagonal blocks of $\m{G}$ at level $\ell$ have numerical rank at most $r_\ell$.  In practice, an adaptive procedure is used to find the ranks.  Finally, we assume that $\m{G}$ is symmetric, since if the trace of nonsymmetric $\m{G}$ is required we can always instead consider a symmetrized form with the same trace such as $\frac{1}{2}(\m{G}+\m{G}^T)$.

\subsubsection{First level of peeling algorithm}\label{sec:level1}

To begin, at level $\ell=1$ the domain is partitioned into four subdomains $\Omega_{1;i}$ with corresponding index sets $\peelidx{1}{i},$ $i=1,\dots,4$ as in \cref{fig:rankstructure}.  We follow the style of Lin \etal \cite{Lin2011} and write the off-diagonal blocks at this level as $\G{1}{ij} \equiv \m{G}\left(\peelidx{1}{i},\peelidx{1}{j}\right)$
to make our notation less cumbersome.

To construct randomized low-rank approximations of $\G{1}{ij}$ for $i\ne j$, we need to find the action of these off-diagonal blocks on random matrices as described in \cref{sec:randomlra}.  Define the block-sparse matrices
\begin{align*}
\m{W}^{(1)}_1 &\equiv \left[\begin{array}{c}\W{1}{1}\\0\\0\\0\end{array}\right],\;\m{W}^{(1)}_2 \equiv \left[\begin{array}{c}0\\\W{1}{2}\\0\\0\end{array}\right],\;
\m{W}^{(1)}_3 \equiv \left[\begin{array}{c}0\\0\\\W{1}{3}\\0\end{array}\right],\; \m{W}^{(1)}_4 \equiv \left[\begin{array}{c}0\\0\\0\\\W{1}{4}\end{array}\right],
\end{align*}
where $\W{1}{j}$ is a random matrix of dimension $|\peelidx{1}{j}| \times (r_1+c)$ for $j=1,\dots,4$.  Applying $\m{G}$ to $\m{W}_1^{(1)}$ gives the action of $\G{1}{1j}$ on the random matrix $\W{1}{1}$ for $j=1,\dots,4$, since
\begin{align*}
\left[\begin{array}{cccc}
\G{1}{11} & \G{1}{12} & \G{1}{13} & \G{1}{14}\\
\G{1}{21} & \G{1}{22} & \G{1}{23} & \G{1}{24}\\
\G{1}{31} & \G{1}{32} & \G{1}{33} & \G{1}{34}\\
\G{1}{41} & \G{1}{42} & \G{1}{43} & \G{1}{44}
\end{array} \right]
\left[\begin{array}{c} \W{1}{1} \\  \\  \\ \\
\end{array}\right] =
\left[\begin{array}{c} \G{1}{11}\W{1}{1}\\
\G{1}{21}\W{1}{1}\\
\G{1}{31}\W{1}{1}\\
\G{1}{41}\W{1}{1}
\end{array}\right].
\end{align*}
The top block of the right-hand side vector above is unused as it is involves a diagonal block of $\m{G}$.   However, the remaining blocks are exactly the matrices $\G{1}{i1}\W{1}{1}$ for $i\ne 1$ as required by the randomized low-rank approximation of \cref{sec:randomlra}.  Applying $\m{G}$ to each $\m{W}^{(1)}_j$ for $j=1,\dots,4$, for each block $\G{1}{ij}$ with $i\ne j$ we obtain a random sampling of its column space $\G{1}{ij}\W{1}{j}$.  Note that by symmetry of $\m{G}$ we also obtain a random sampling of the row space of each block since $\G{1}{ij}\W{1}{j} = \G{1}{ji}^T\W{1}{j}$.

Using \eqref{eq:lowrankapprox} to construct rank-$r_1$ approximations of each of these blocks, we write the approximation of $\G{1}{ij}$ as $\G{1}{ij} \approx \Ghat{1}{i}{j} \equiv \U{1}{ij} \M{1}{ij} \U{1}{ji}^T,$
where the approximation is accurate to the specified tolerance $\epsilon_\text{peel}$ with high probability.

Defining the matrix $\m{G}^{(1)}\in\R^{n\times n}$ with blocks given by
\begin{align*}
\m{G}^{(1)}\left(\peelidx{1}{i}, \peelidx{1}{j}\right) = \left\{\begin{array}{ll} \Ghat{1}{i}{j} & i\ne j,  \\ 0 & \text{else},\end{array}\right.
\end{align*}
we obtain
\begin{align*}
\m{G} - \m{G}^{(1)} &\equiv \m{G} - \left[\begin{array}{cccc}
& \Ghat{1}{1}{2} & \Ghat{1}{1}{3}  & \Ghat{1}{1}{4}\\
\Ghat{1}{2}{1} &  & \Ghat{1}{2}{3} & \Ghat{1}{2}{4}\\
\Ghat{1}{3}{1} & \Ghat{1}{3}{2} &  & \Ghat{1}{3}{4}\\
\Ghat{1}{4}{1} & \Ghat{1}{4}{2} & \Ghat{1}{4}{3} &
\end{array} \right]\approx \left[\begin{array}{cccc}
\G{1}{11} &  &  & \\
 & \G{1}{22} &  & \\{}
 & & \G{1}{33} & \\
 & &  & \G{1}{44}
\end{array} \right].
\end{align*}
In other words, we have approximated the off-diagonal blocks at this level to a specified accuracy and used the result to obtain a fast operator $\m{G}-\m{G}^{(1)}$ that is block-diagonal with diagonal blocks the same as those of $\m{G}$.

\begin{remark}
We note that the matrix $\m{G}^{(1)}$ is not explicitly assembled as a dense matrix inside the peeling algorithm.  Instead, we store the non-zero blocks in low-rank form so that $\m{G}^{(1)}$ may be efficiently applied to vectors.
\end{remark}

\subsubsection{Second level of peeling algorithm}\label{sec:level2}
In the next step of the peeling algorithm, we recurse on the diagonal subblocks $\G{1}{ii}$ for $i=1,\dots,4$.  Partitioning each subdomain $\Omega_{1;i}$ at level $\ell=1$ into four child subdomains at level $\ell=2$ using the quadtree structure and renumbering blocks accordingly, we write the diagonal blocks $\G{1}{ii}$ for $i=1,\dots,4$ as
\begin{align*}
\G{1}{11} &= \left[\begin{array}{cccc}
\G{2}{11} & \G{2}{12} & \G{2}{13} & \G{2}{14}\\
\G{2}{21} & \G{2}{22} & \G{2}{23} & \G{2}{24}\\
\G{2}{31} & \G{2}{32} & \G{2}{33} & \G{2}{34}\\
\G{2}{41} & \G{2}{42} & \G{2}{43} & \G{2}{44}
\end{array} \right], \;
\G{1}{22} = \left[\begin{array}{cccc}
\G{2}{55} & \G{2}{56} & \G{2}{57} & \G{2}{58}\\
\G{2}{65} & \G{2}{66} & \G{2}{67} & \G{2}{68}\\
\G{2}{75} & \G{2}{76} & \G{2}{77} & \G{2}{78}\\
\G{2}{85} & \G{2}{86} & \G{2}{87} & \G{2}{88}
\end{array} \right],
\end{align*}
and so on for $\G{1}{33}$ and $\G{1}{44}$.

For each $j=1,\dots,16$ we define the random matrix $\W{2}{j}\in\R^{|\peelidx{2}{j}|\times(r_{2}+c)}$, which is appropriately sized to give a random sample of the column space of $\G{2}{ij}$ for each $i=1,\dots,16$, $i\ne j$.  We can minimize the number of times we apply the operator $\m{G}-\m{G}^{(1)}$ as follows due to its block diagonal structure.  For each $k=1,\dots,4$, we define $\m{W}^{(2)}_k\in\R^{n\times (r_{2}+c)}$ to have rows divided into 16 blocks according to
\begin{align*}
\m{W}^{(2)}_k(\peelidx{2}{j},:) = \left\{\begin{array}{ll} \W{2}{j} & j\in\{k,k+4,k+8,k+12\}, \\ 0 & \text{else.}\end{array}\right.
\end{align*}
In other words, block $k$ of $\m{W}^{(2)}_k$ is nonzero, as is every fourth block after $k$.

\begin{definition}[Quadtree siblings]\label{def:sibs}
In the context of the quadtree decomposition of $\Omega$, we say that $\Omega_{\ell;i}$ and $\Omega_{\ell;j}$ are \emph{siblings}  if $i\ne j$ and both $\Omega_{\ell;i}\subset\Omega_{\ell-1;k}$ and $\Omega_{\ell;j}\subset\Omega_{\ell-1;k}$ for some $1\le k \le 4^{(\ell-1)}$.
\end{definition}

Let $\m{B}_k\equiv(\m{G} - \m{G}^{(1)})\m{W}^{(2)}_k$ and suppose that $\m{W}^{(2)}_k(\peelidx{2}{j},:)$ is nonzero.  For each $i$ such that $\Omega_{2;i}$ and $\Omega_{2;j}$ are siblings, we have $\m{B}_k(\peelidx{2}{i},:) \approx  \G{2}{ij}\W{2}{j}.$
For example, in $\m{W}^{(2)}_1$ the nonzero blocks are $\m{W}^{(2)}_1(\peelidx{2}{j},:) = \W{2}{j}$ for $j\in\{1,5,9,13\}$, so 
\begin{align*}
\m{B}_1(\peelidx{2}{i},:)  \approx \left\{\begin{array}{cl} \G{2}{i1}\W{2}{1}& i=1,\dots,4,\\
\G{2}{i5}\W{2}{5}& i=5,\dots,8,\\
\G{2}{i9}\W{2}{9}& i=9,\dots,12,\\
\G{2}{i13}\W{2}{13}& i=13,\dots,16.\end{array} \right.
\end{align*}
Therefore, applying $\m{G}-\m{G}^{(1)}$ to $\m{W}^{(2)}_k$ for $k=1,\dots,4$ gives a random sample of the column space and row space of $\G{2}{ij}$ for each $i$ and $j$ such that $\Omega_{2;i}$ and $\Omega_{2;j}$ are siblings.  For all such $i$ and $j$ we use the randomized low-rank approximation algorithm as before to construct
\begin{align*}
\Ghat{2}{i}{j}= \U{2}{ij} \M{2}{ij} \U{2}{ji}^T.
\end{align*}  
Defining $\m{G}^{(2)}\in\R^{n\times n}$ with blocks
\begin{align*}
\m{G}^{(2)}\left(\peelidx{2}{i}, \peelidx{2}{j}\right) = \left\{\begin{array}{ll} \Ghat{2}{i}{j} & \text{if $\Omega_{2;i}$ and $\Omega_{2;j}$ are siblings},  \\ 0 & \text{else},\end{array}\right.
\end{align*}
we have that $\m{G}-\m{G}^{(1)} - \m{G}^{(2)}$ is approximately block-diagonal with diagonal blocks $\G{2}{ii}$ for $i=1,\dots,16$.

\subsubsection{Subsequent levels of peeling algorithm}
In general at level $\ell>2$ we see that $\m{G} - \sum_{m=1}^{\ell-1} \m{G}^{(m)}$ is approximately block-diagonal with $4^{(\ell-1)}$ diagonal blocks.  For each $k=1,\dots,4$ we define $\m{W}^{(\ell)}_k\in\R^{n\times (r_{\ell}+c)}$ to have rows divided into $4^\ell$ blocks according to
\begin{align*}
\m{W}^{(\ell)}_k(\peelidx{\ell}{j},:) = \left\{\begin{array}{ll} \W{\ell}{j} & j\equiv k \text{ (mod 4)}, \\ 0 & \text{else,}\end{array}\right.
\end{align*}
where each $\W{\ell}{j}$ is a random matrix of size $\R^{|\peelidx{\ell}{j}|\times r_\ell}.$  Using the same logic as in \cref{sec:level2}, we apply $\m{G} - \sum_{m=1}^{\ell-1}\m{G}^{(m)}$ to $\m{W}^{(\ell)}_k$ for each $k=1,\dots,4$ and use the results to construct low rank approximations
\begin{align*}
\Ghat{\ell}{i}{j}= \U{\ell}{ij} \M{\ell}{ij} \U{\ell}{ji}^T.
\end{align*}
for each $i$ and $j$ such that $\Omega_{\ell;i}$ and $\Omega_{\ell;j}$ are siblings.  We define
\begin{align*}
\m{G}^{(\ell)}\left(\peelidx{\ell}{i}, \peelidx{\ell}{j}\right) = \left\{\begin{array}{ll} \Ghat{\ell}{i}{j} & \text{if $\Omega_{\ell;i}$ and $\Omega_{\ell;j}$ are siblings},  \\ 0 & \text{else}\end{array}\right.
\end{align*}
such that $\m{G} - \sum_{m=1}^{\ell}\m{G}^{(m)}$ is approximately block-diagonal with $4^\ell$ diagonal blocks.

\subsubsection{Extracting the trace}
At the bottom level of the quadtree, each diagonal block of $\m{G}-\sum_{m=1}^{L} \m{G}^{(m)}$ is of a constant size independent of $n$ as discussed in \cref{sec:blockcompression}.  Define $\nn{L}{i}\equiv |\peelidx{L}{i}|$ and $n_L \equiv \max_i \,\nn{L}{i}$ such that $n_L$ is the maximum number of observations in a leaf-level subdomain. We construct a block matrix $\m{E}\in\R^{n\times n_L}$ such that
\begin{align*}
\m{E}\left(\peelidx{L}{i},\left[\nn{L}{i}\right]\right) = \m{I}\in\R^{|\peelidx{L}{i}|\times|\peelidx{L}{i}|}
\end{align*}
for each $i$, where $\m{I}$ is an appropriately-sized identity matrix. Letting 
\begin{align*}
\m{H}\equiv\left(\m{G}-\sum_{m=1}^{L} \m{G}^{(m)}\right)\m{E}, 
\end{align*} we find that $\m{H}\left(\peelidx{L}{i},\left[\nn{L}{i}\right]\right) \approx \G{L}{ii}$
for each $i = 1,\dots,4^L$.  We can then approximate the trace of $\m{G}$ using the relation
\begin{align*}
\Tr(\m{G}) = \sum_{i=1}^{4^{L}} \Tr(\m{G}_{L;ii}) \approx \sum_{i=1}^{4^L} \Tr(\m{H}\left(\peelidx{L}{i},\left[\nn{L}{i}\right]\right)).
\end{align*}

\begin{remark}
When using the peeling algorithm to construct an approximate trace of an operator $\m{G}$ with numerically low-rank off-diagonal blocks, it is important to note that we do not have direct control of the relative error of the trace approximation.   This is because a matrix with diagonal entries with large absolute value but mixed signs can have a small trace due to cancellation.  In practice, however, our numerical results in \cref{sec:results} show excellent agreement between the approximate trace and true trace.
\end{remark}

\subsection{Computational complexity}\label{sec:peelcomplexity}

For each level of the weak-admissibility-based peeling algorithm described  in \cref{sec:peelsub} there are two key steps: applying the operator $\m{G} - \sum_{m=1}^{\ell-1} \m{G}^{(m)}$ and forming the low-rank factorizations $\Ghat{\ell}{i}{j}$ for each $i$ and $j$ such that $\Omega_{\ell;i}$ and $\Omega_{\ell;j}$ are siblings.  Analyzing the cost of these steps leads to the following complexity result.

\begin{theorem}\label{thm:peelcomplexity}
Let the cost of applying $\m{G}\in\R^{n\times n}$ to a vector be $T_\text{apply}$ and assume that the observations are uniformly distributed in $\Omega$ such that $|\peelidx{\ell}{i}| = \O(4^{-\ell}n)$ for each $1\le\ell\le L$ and $1\le i \le 4^\ell$.  Assuming that the ranks of the off-diagonal blocks $\G{\ell}{i}{j}$ are bounded by $r_\ell$ for each $i$ and $j$ such that $\Omega_{\ell;i}$ and $\Omega_{\ell;j}$ are siblings, and define 
\begin{align*}
s_1\equiv \sum_{\ell=1}^L r_\ell, \quad \text{and} \quad s_2\equiv \sum_{\ell=1}^L r_\ell^2.
\end{align*}
Then the complexity of the weak-admissibility-based peeling algorithm is
\begin{align}\label{eq:peeltime}
T_\text{peel}=\O(T_\text{apply}s_1 +ns_2\log n).
\end{align}
The storage complexity is $\O(ns_1)$.
\end{theorem}
\begin{proof}
We adapt the proof of Lin \etal \cite{Lin2011} to the weak admissibility case.  At the first level, applying $\m{G}$ to each $\m{W}^{(1)}_k$ costs $\O(T_\text{apply}r_1)$ and each randomized factorization costs $\O(nr_1^2)$, leading to an overall cost for level 1 of $\O(T_\text{apply}r_1 + nr_1^2)$.

 At level $\ell > 1$, we break the cost of applying $\m{G}-\sum_{m=1}^{\ell-1}\m{G}^{(m)}$ into two pieces.  The cost of applying $\m{G}$ to each $\m{W}^{(\ell)}_k$ is $\O(T_\text{apply}r_\ell)$.  The matrix $\sum_{m=1}^{\ell-1}\m{G}^{(m)}$ is a heavily structured matrix with blocks in low-rank form.  Applying this to each $\m{W}^{(\ell)}_k$ costs $\O\left(\sum_{m=1}^{\ell-1} nr_mr_\ell\right)$, which is $\O(ns_2)$.  We additionally must construct each randomized factorization at this level.  Each one costs $\O(4^{-\ell}nr_\ell^2)$ and there are $\O(4^\ell)$ off-diagonal blocks to compress at this level, so the overall cost for level $\ell$ is $\O(T_\text{apply}r_\ell + ns_2+nr_\ell^2)$.

 Summing the cost of each level from $\ell=1,\dots,L$, we obtain \eqref{eq:peeltime}.  Note that at level $\ell=L$, we must additionally extract the diagonal blocks, but by the assumption these blocks are of constant size so this does not increase the asymptotic cost.  The storage complexity comes from noting that at level $\ell$ we must store the $\O(4^\ell)$ matrices of rank $r_\ell$, where each has outer dimension that is $\O(4^{-\ell}n).$
\end{proof}

When the underlying matrix has the rank of all off-diagonal blocks  bounded by $r_\ell = \O(1)$ for all $\ell$, then the computational complexity of weak peeling is $\softO(T_\text{apply} +n)$, where we use the so-called ``soft-O'' notation from theoretical computer science to suppress factors that are polylogarithmic in $n$.  {In this case,} peeling $\m{G}=\m{\Sigma}$ itself using its recursive skeletonization factorization results in $\softO(n)$ complexity for both time and memory{.}

Many real matrices of interest, however, do not exhibit off-diagonal blocks with ranks independent of $n$.  For example, our experiments with the Mat\'ern kernel of \eqref{eq:matern} show that a constant number of off-diagonal blocks at each level of the hierarchy exhibit ranks bounded only as $r_\ell = \O(2^{-\ell}\sqrt{n})$.  This coincides with the argument for rank growth in the recursive skeletonization factorization in \cref{sec:rskelfscaling}.  Thus, this simplified peeling algorithm in the case of the Mat\'ern kernel has asymptotic time complexity $\softO(n^2)$ and storage complexity $\softO(n^{3/2})$, where we pick up at most a polylogarithmic factor in the ranks since we are looking at $\m{G}=\m{\Sigma}^{-1}\m{\Sigma}_i$ {and not $\m{\Sigma}$ itself}.  In theory, using the simple peeling algorithm described here is asymptotically no better than extracting the trace by applying $\m{G}$ to the coordinate vectors $e_i$ for $i=1,\dots,n$.  This necessitates the standard form of peeling for large problems. 

\begin{remark}
In practice, the standard form of the peeling algorithm \cite{Lin2011} that uses the full generality of strong admissibility can be employed to remedy such rank growth by explicitly avoiding compression of off-diagonal blocks that are not sufficiently low-rank.  Using the modifications described in that paper, the complexity of peeling follows the same bound as \cref{thm:peelcomplexity} but with the rank bound $r_\ell$ referring to a bound on the ranks of only those blocks that are compressed in the strongly-admissible hierarchical format.  {We find in \cref{sec:peelingruntime} that using peeling based on strong admissibility is more efficient when $n$ is large, as expected.  However, the implicit constants in the asymptotic runtime lead to weak admissibility being more efficient for moderately-sized problems.}
\end{remark}

 We summarize our complexity results in \cref{tab:peeling}.  Note that these results were derived on the assumption that $n_\ell=\O(4^{-\ell}n)$, \ie, a quasi-uniform distribution of observations and a perfect quadtree decomposition of space.  In practice observations that are distributed in a different fashion can actually exhibit better behavior, particularly if the observations are concentrated around a quasi-one-dimensional curve \cite{hifie}.

\setlength\extrarowheight{2pt}
\begin{table}
\ra{1.0}
\caption{The runtime and storage complexity of the peeling algorithm depend on the asymptotic rank $r_\ell$ of off-diagonal blocks of $\m{G}$ at level $\ell$.  The tabulated complexities are based on the assmption that the recursive skeletonization factorization is used to apply $\m{G}$ as a fast operator.\label{tab:peeling}}
\centering
\begin{tabular}{lll} \toprule
$r_\ell$ & Time & Storage \\ \midrule
$\O(\log n_\ell)$  &$\softO(n)$   & $\softO(n)$  \\
$\O(\sqrt n_\ell)$ &$\softO\left(n^2\right)$& $\softO(n^{3/2})$ \\ \bottomrule
\end{tabular}
\end{table}

\section{Summary of MLE framework}\label{sec:summary}
In \cref{alg:mle} we  summarize our complete framework for computing the log-likelihood $\ell(\theta)$ and gradient $g(\theta)$ given $\theta$, which can be used inside of any first-order optimization routine for Gaussian process maximum likelihood estimation.  As mentioned previously, the approach is flexible and does not rely on the specific hierarchical factorization used (\eg, the recursive skeletonization factorization, the hierarchical interpolative factorization, the strong recursive skeletonization factorization) or the form of peeling used (\ie, the peeling based on weak admissibility described in \cref{sec:peel} or the form by Lin \etal \cite{Lin2011} based on strong admissibility).  Rather, the exact components of the framework should be decided on a case-by-case basis depending on the rank properties of the kernel family.

\begin{algorithm}
 \small
  \caption{Computing the Gaussian process log-likelihood and gradient}
  \label{alg:mle}
  \begin{algorithmic}[1]
   \Statex Given: observation vector $z\in\R^n$, observation locations $\{x_i\}_{i=1}^n \subset \R^2$, peel tolerance $\epsilon_\text{peel}$, factorization tolerance $\epsilon_\text{fact} < \epsilon_\text{peel}$, parameter vector $\theta\in\R^p$, and covariance kernel $K(\cdot,\cdot;\theta)$
   \State{\texttt{// Factor $\m{\Sigma}$ with hierarchical factorization}}
   \State $\m{F} \gets$ Recursive skeletonization factorization of $\m{\Sigma}$ with tolerance $\epsilon_\text{fact}$
   \State{\texttt{// Use fast hierarchical solve and log-determinant}}
   \State $\hat\ell(\theta) \gets -\frac{1}{2}z^T\m{F}^{-1}z - \frac{1}{2}\log |\m{F}| - \frac{1}{2}\log 2\pi \approx \ell(\theta)$
   \For{$i = 1,\dots, p$}
	   \State{\texttt{// Factor $\m{\Sigma}_i$ with hierarchical factorization}}
	   \State $\m{F}_i \gets$ Recursive skeletonization factorization of $\m{\Sigma}_i$ with tolerance $\epsilon_\text{fact})$
	   \State{\texttt{// Compute trace of $\m{\Sigma}^{-1}\m{\Sigma}_i$ with peeling algorithm}}
	   \State $t_i \gets$ Trace of operator $\frac{1}{2}(\m{F}^{-1}\m{F}_i + \m{F}_i\m{F}^{-1})$ via peeling algorithm with tolerance $\epsilon_\text{peel}$
	   \State{\texttt{// Use fast hierarchical apply and solve}}
	   \State $\hat g_i \gets \frac{1}{2}z^T\m{F}^{-1}\m{F}_i\m{F}^{-1}z - \frac{1}{2}t_i \approx g_i$
   \EndFor
   \Statex Output: $\hat \ell(\theta)$ and $\hat g(\theta)$
  \end{algorithmic}
 \end{algorithm}

 {
 \begin{remark}\label{remark:conditional}
After estimation of the parameter vector $\theta$, there remains the question of how to sample from the Gaussian process conditioned on the observed data $z$.  Assuming $[z',z]^T$ is jointly distributed according to the original Gaussian process, this conditional distribution is given by
\begin{align*}
z'|z \sim N(\m{\Sigma}_{12}\m{\Sigma}_{22}^{-1}z,\, \m{\Sigma}_{11}-\m{\Sigma}_{12}\m{\Sigma}_{22}^{-1}\m{\Sigma}_{12}^T),
\end{align*}
where $\m{\Sigma}_{11}$ is the covariance matrix of $z'$, $\m{\Sigma}_{22}$ is the covariance matrix of $z$, and so on.  Using the identity
\begin{align*}
\m{\Sigma}_{11}-\m{\Sigma}_{12}\m{\Sigma}_{22}^{-1}\m{\Sigma}_{12}^T &=\left[\begin{array}{cc} \m{I} & -\m{\Sigma}_{12}\m{\Sigma}_{22}^{-1} \end{array}\right]\left[\begin{array}{cc}\m{\Sigma}_{11} & \m{\Sigma}_{12}\\ \m{\Sigma}_{12}^T & \m{\Sigma}_{22} \end{array}\right]\left[\begin{array}{c} \m{I}\\ -\m{\Sigma}_{22}^{-1}\m{\Sigma}_{12}^T  \end{array}\right]
\end{align*}
and letting $\m{\Sigma}$ denote the two-by-two block matrix in a slight abuse of notation, we can apply a square-root of $\m{\Sigma}_{11}-\m{\Sigma}_{12}\m{\Sigma}_{22}^{-1}\m{\Sigma}_{12}^T$ with skeletonization factorizations $\m{F}_{22}\approx\m{\Sigma}_{22}$ and $\m{F}\approx \m{\Sigma}$ by using $\m{F}^{1/2}$ to apply a square-root of $\m{\Sigma}$, $\m{F}$ to apply $\m{\Sigma}_{12}$ through appropriate padding, and $\m{F}^{-1}_{22}$ to apply $\m{\Sigma}_{22}^{-1}$.  This gives a fast method for sampling from the conditional distribution or computing the conditional mean.
\end{remark}
}

\section{Numerical results}\label{sec:results}

To demonstrate the effectiveness of our approach to Gaussian process maximum likelihood estimation, we first test the accuracy and runtime of the peeling-based technique for approximating the trace and then test our full method on two examples using synthetic datasets and one example using a dataset of measurements of ocean surface temperatures.  For examples, we take the number of proxy points to be $n_\text{prox}=256,$ and use a {quadtree} decomposition of space with a maximum of $n_\text{occ}=64$ points per leaf subdomain.

In our tests we use the FLAM library (\url{https://github.
com/klho/FLAM/}) for the recursive skeletonization factorization and a custom implementation of matrix peeling as described in \cref{sec:peelsub}. This additional code is {available} at \url{https://github.com/asdamle/GPMLE/}.  All numerical results shown were run in MATLAB\textsuperscript{\textregistered} R2015a on a quad-socket Intel\textsuperscript{\textregistered} Xeon\textsuperscript{\textregistered} E5-4640 processor clocked at 2.4 GHz using up to 1.5 TB of RAM.

\subsection{Runtime scaling of the peeling algorithm}\label{sec:peelingruntime}
To begin, we investigate the numerical performance of the peeling algorithm on synthetic examples.  We take the observation locations $\{x_i\}_{i=1}^n$ to be a $\sqrt{n} \times \sqrt{n}$ grid of points uniformly discretizing the square $[0,100]^2\subset \R^2.$  We let $\theta=[\theta_1,\theta_2]$ parameterize the correlation length scale of the process in each coordinate direction, defining the scaled distance
\begin{align*}
\|x-y\|^2_\theta = \frac{(x_1-y_1)^2}{\theta_1^2} + \frac{(x_2-y_2)^2}{\theta_2^2},
\end{align*}
where here $x_i$ and $y_i$ are used to denote components of vectors $x$ and $y$.  Using this parameterization and incorporating an additive noise term, the two kernels we test are the {rational quadratic kernel of \eqref{eq:rationalquadratic} with $\alpha=1/2$,
\begin{align}\label{eq:specrq}
K_{RQ}(x,y;\theta) &= \left(1+\|x-y\|^2_\theta\right)^{-1/2} + \sigma_N^2 \delta_{xy},
\end{align}
}
and the Mat\'ern kernel of \eqref{eq:matern} with parameter $\nu=3/2$,
\begin{align}\label{eq:specmatern}
K_M(x,y;\theta) &= (1 + \sqrt{3}\|x-y\|_\theta)\exp(-\sqrt{3}\|x-y\|_\theta) + \sigma_N^2 \delta_{xy}.
\end{align}
Here $\delta_{xy}$ is the Kronecker {delta,} which satisfies $\delta_{xy}=1$ if $x=y$ and $\delta_{xy}=0$ otherwise.

\begin{remark}
In both \eqref{eq:specrq} and \eqref{eq:specmatern} the additional term $\sigma_N^2 \delta_{xy}$ can be interpreted as modeling additive white noise with variance $\sigma_N^2$ on top of the base Gaussian process model.  In practice, this so-called ``nugget effect'' is frequently incorporated to account for measurement error or small-scale variation from other sources \cite{matheron} and, further, is numerically necessary for many choices of parameter $\theta$ due to exceedingly poor conditioning of many kernel matrices.
\end{remark}
We compute high-accuracy recursive skeletonization factorizations of the matrices $\m{\Sigma}$ and $\m{\Sigma}_1\equiv\frac{\partial}{\partial\theta_1}\m{\Sigma}$, which we combine to obtain the fast black-box operator
\begin{align}\label{eq:Gbb}
\m{G} = \frac{1}{2}(\m{\Sigma}^{-1}\m{\Sigma}_1 + \m{\Sigma}_1\m{\Sigma}^{-1})
\end{align}
for input to the peeling algorithm to compute the trace to specified tolerance $\epsilon_\text{peel}=1\times 10^{-6}$.  We choose the parameter vector $\theta=[10,7]$ for these examples as in \cref{fig:gpdiff} (left), and set the noise parameter at $\sigma_N^2=1\times 10^{-4}$.

\begin{figure}%
\centering
\begin{minipage}{0.4\textwidth}%
\centering
\includegraphics[width=\textwidth]{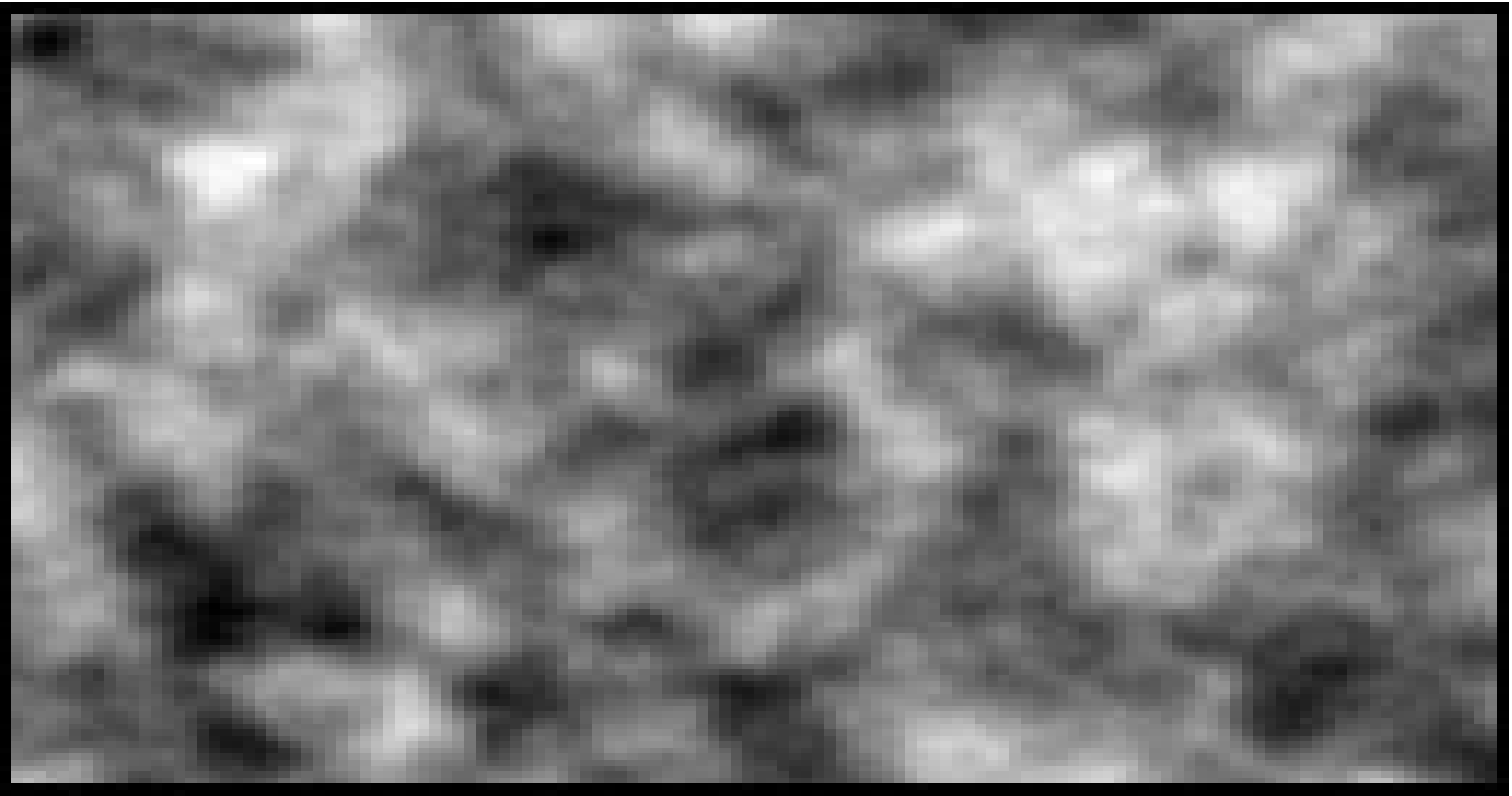}
\end{minipage}%
\quad\quad\quad\quad
\begin{minipage}{0.4\textwidth}%
\centering
\includegraphics[width=\textwidth]{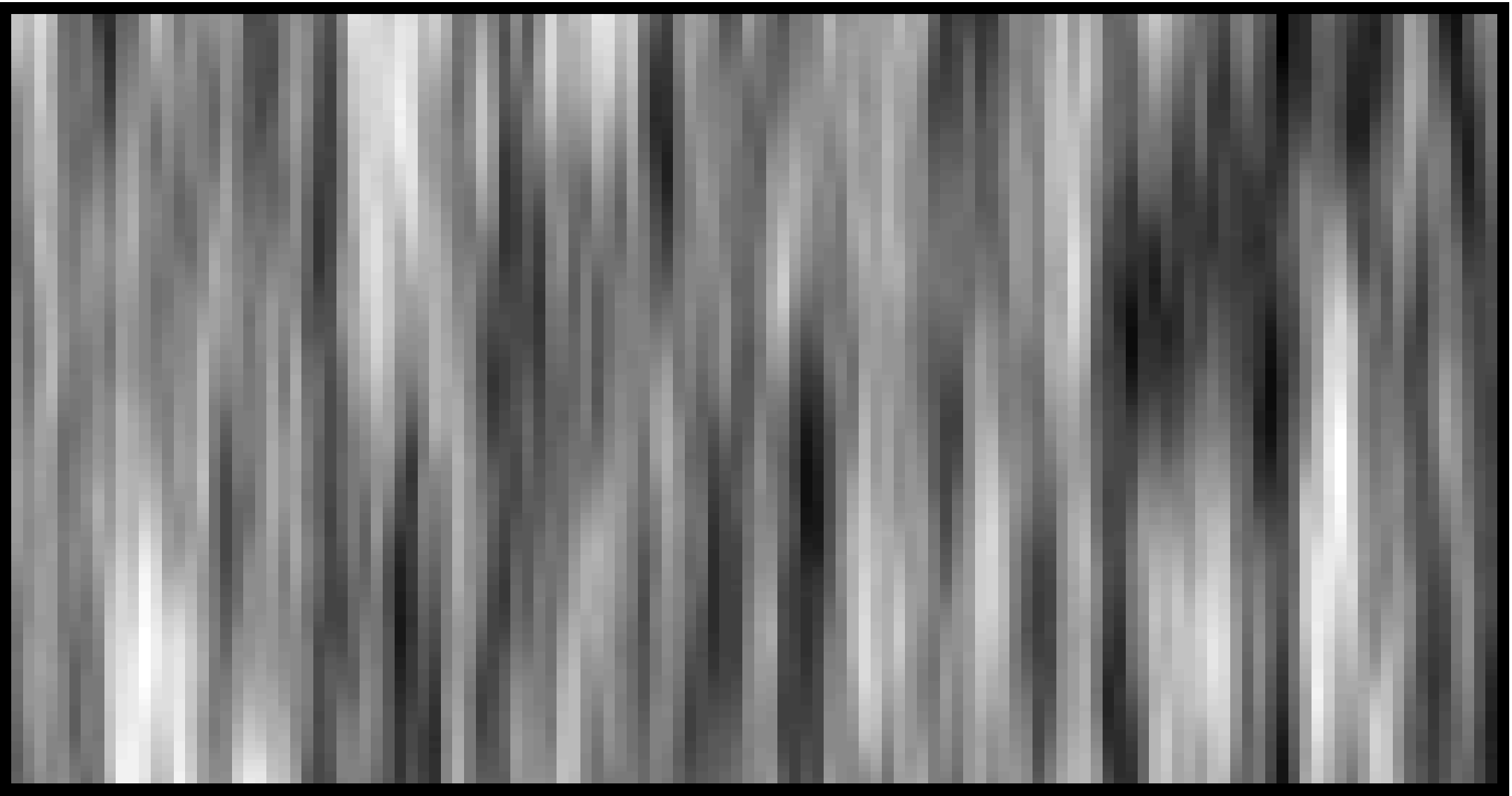}
\end{minipage}%
\caption{ Two different realizations on the domain $[0,200]\times[0,100]$ of the Mat\'ern kernel Gaussian process with covariance seen in \eqref{eq:specmatern} and noise parameter $\sigma_N^2=0$.  In the left figure the parameter vector is $\theta=[10,7]$ corresponding to a kernel that is relatively close to isotropic.  In contrast, in the right figure the parameter vector $\theta=[3,30]$ generates strong anisotropy.\label{fig:gpdiff}}
\end{figure}

Beginning with the {rational quadratic} kernel, in \cref{tab:peelingresults_rq} we give runtime results for both the simplified peeling algorithm described in \cref{sec:peelsub} (``weak peeling'') as well as the full strong-admissibility-based peeling algorithm of Lin \etal \cite{Lin2011} (``strong peeling'').  As can be seen in \cref{fig:peel_prof} (left), the runtime of the peeling algorithm with the kernel \eqref{eq:specrq} seems to scale between $\O(n)$ and $\O(n^{3/2})$ with the number of observations $n$, regardless of whether weak or strong peeling is used.  Further, the relative error in the trace approximation, $e_\text{peel}$ is {near the specified tolerance $\epsilon_\text{peel}$, though the tolerance is not a hard upper bound.}  Note that we omit the relative error for our largest example, as the operator was too large to determine the true trace using the na\"ive approach.

\begin{table}
\caption{Runtime $t_\text{peel}$ of the the peeling algorithm with the {rational quadratic} kernel of \eqref{eq:specrq}.  Note that we omit the relative error $e_\text{peel}$  in the estimated trace for our largest example, as the operator was too large to determine the true trace using the na\"ive approach.\label{tab:peelingresults_rq}}
\centering
\ra{1.0}
\begin{tabular}{ccccc} \toprule
$n$ & $t_\text{peel,weak}$ (s)& $e_\text{peel,weak}$  & $t_\text{peel,strong}$ (s) & $e_\text{peel,strong}$ \\ \midrule
 $64^2$ & $9.17\times 10^0$ & $5.68\times 10^{-7}$ & $4.85\times10^1$ & $1.60\times 10^{-7}$ \\ 
 $128^2$& $9.16\times 10^1$ & $1.02\times 10^{-5}$ & $5.64\times 10^2$ & $2.58\times 10^{-7}$ \\ 
 $256^2$& $6.63\times 10^2$ & $3.72\times 10^{-5}$ & $3.04\times 10^3$ & $3.32\times 10^{-6}$ \\
 $512^2$& $2.88\times 10^3$ & - & $1.73\times 10^4$ & - \\ \bottomrule
\end{tabular}
\end{table}

In contrast, the results in \cref{tab:peelingresults_matern} for the Mat\'ern kernel in \eqref{eq:specmatern} show different scaling behavior for weak and strong peeling.  In \cref{fig:peel_prof} (right), we see that the runtime for weak peeling seems to be close to quadratic in the number of observations, which agrees with our analysis from \cref{sec:peel}.  Using strong peeling, however, the complexity of peeling scales considerably better, ultimately following the $\O(n^{3/2})$ trend line.  We see again that the relative trace error is well-controlled by $\epsilon_\text{peel}$ in both cases.

\begin{table}
\caption{Runtime $t_\text{peel}$ of the the peeling algorithm with the Mat\'ern kernel of \eqref{eq:specmatern}.  Note that we omit the relative error $e_\text{peel}$ in the estimated trace for our largest example, as the operator was too large to determine the true trace using the na\"ive approach.\label{tab:peelingresults_matern}}
\centering
\ra{1.0}
\begin{tabular}{ccccc} \toprule
 $n$ & $t_\text{peel,weak}$ (s)& $e_\text{peel,weak}$ & $t_\text{peel,strong}$ (s) & $e_\text{peel,strong}$ \\ \midrule
 $64^2$ & $6.03\times 10^0$  & $5.73\times 10^{-8}$ & $2.06\times 10^1$ & $4.78\times 10^{-10}$ \\
 $128^2$& $5.30\times 10^1$ & $2.46\times 10^{-7}$ & $2.29\times 10^2$ & $3.36\times 10^{-10}$ \\ 
 $256^2$& $5.37\times 10^2$ & $4.28\times 10^{-6}$  & $1.62\times 10^3$ & $8.14\times 10^{-10}$ \\
 $512^2$&  $7.07 \times 10^3$ & - &  $1.00\times 10^4$ &-  \\ \bottomrule
\end{tabular}
\end{table}

Though the observed scaling behavior of strong peeling is as good or better than that for weak peeling for both kernels, in practice we see that for problems with up to a quarter of a million observations weak peeling has a smaller time-to-solution.  As such, in the remainder of our examples we show results using only weak peeling.

\begin{figure}%
\centering
\begin{minipage}{0.45\textwidth}%
\centering
\includegraphics[width=\textwidth]{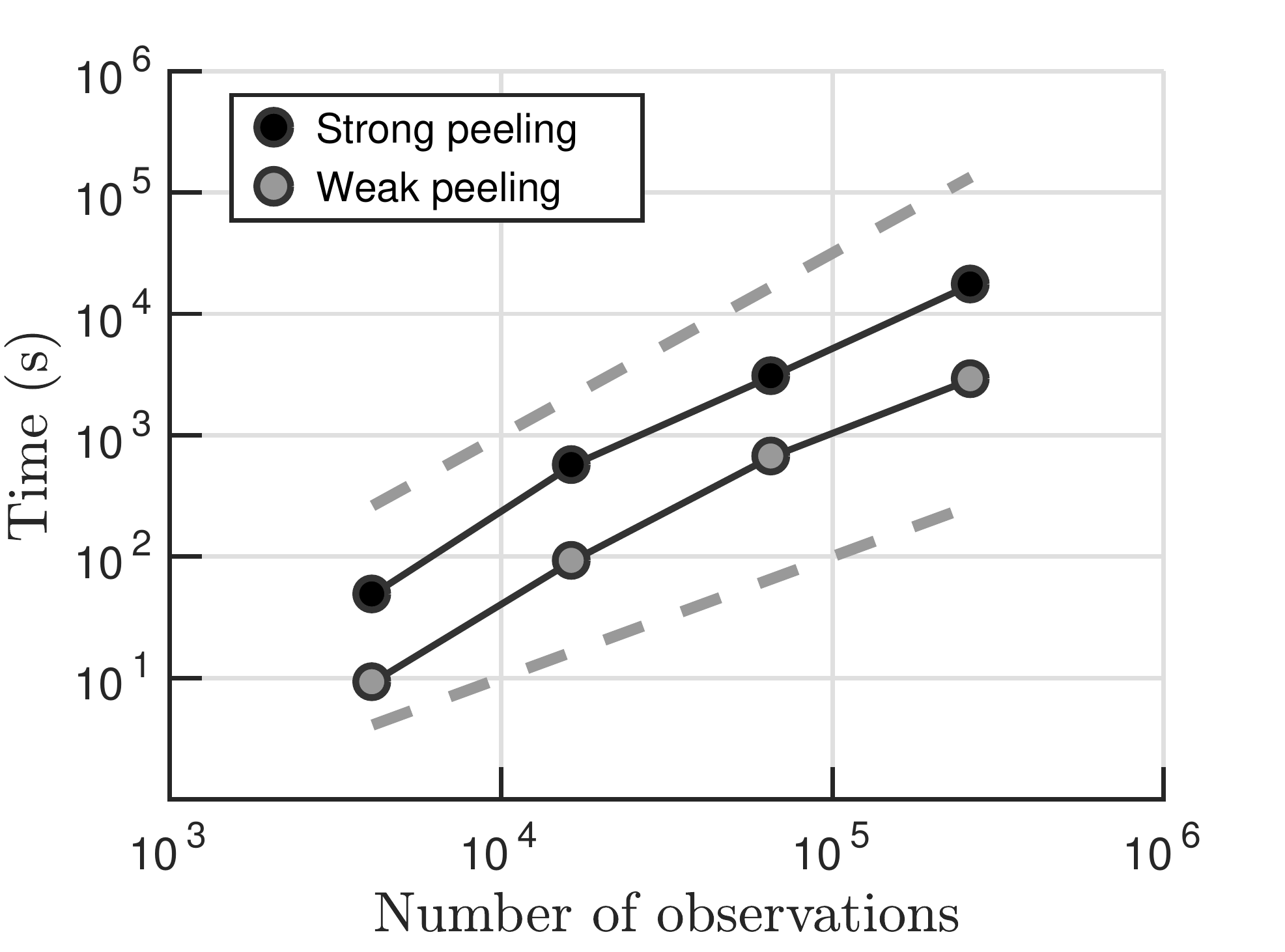}
\end{minipage}%
\quad\quad\quad
\begin{minipage}{0.45\textwidth}%
\centering
\includegraphics[width=\textwidth]{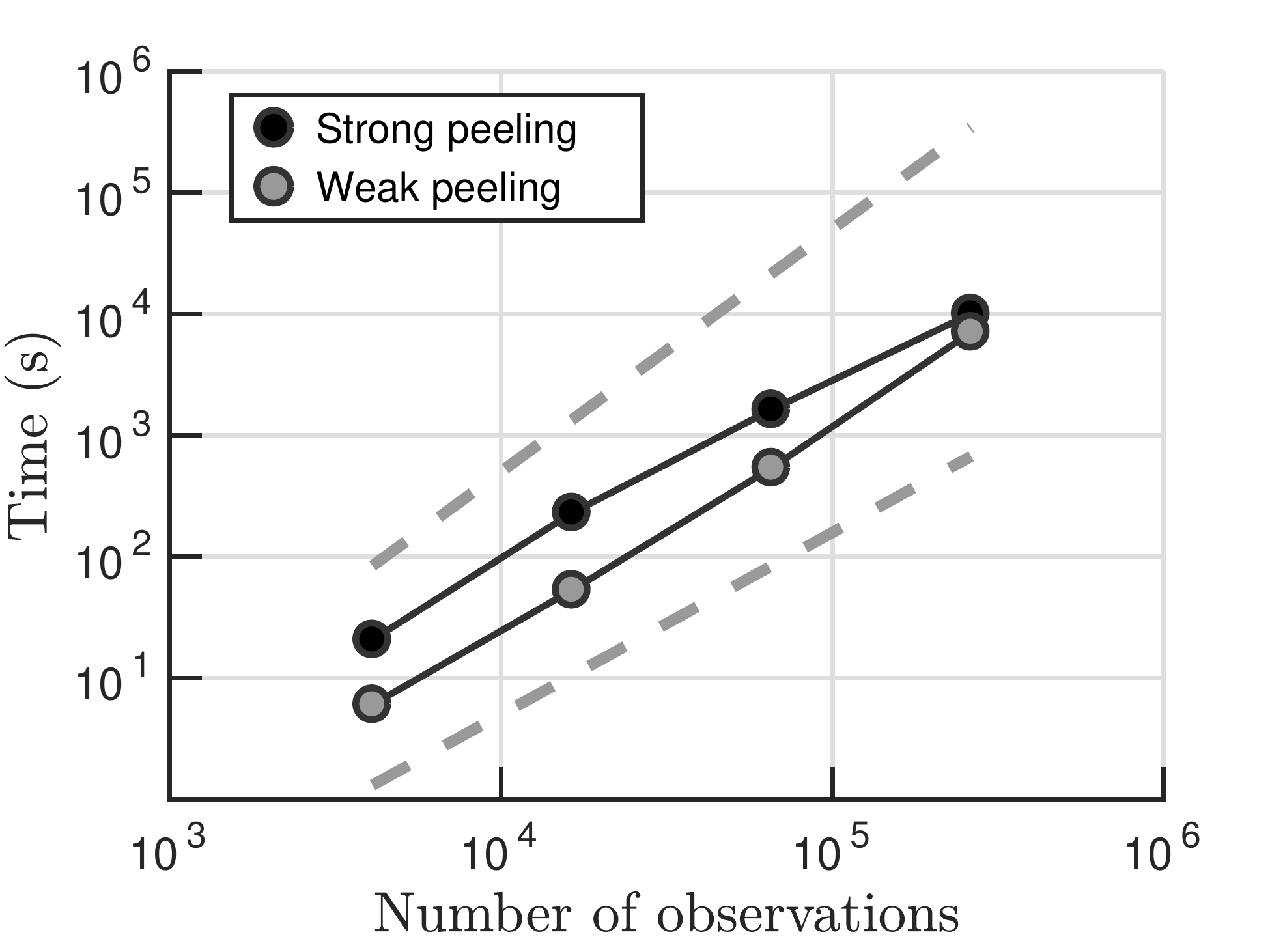}
\end{minipage}%

\caption{On the left the runtime of peeling for the {rational quadratic} kernel is plotted along with a $\O(n^{3/2})$ trend line (top) and a $\O(n)$ trend line (bottom), {showing subquadratic scaling for weak peeling in this case}.  In contrast, on the right the runtime of peeling for the Mat\'ern kernel is plotted along with a $\O(n^2)$ trend line (top) and a $\O(n^{3/2})$ trend line (bottom).  We see that weak peeling with the Mat\'ern kernel seems to ultimately exhibit quadratic scaling, whereas strong peeling seems to exhibit slightly better than $\O(n^{3/2})$ scaling.  {The corresponding data are given in \cref{tab:peelingresults_rq,tab:peelingresults_matern}}. \label{fig:peel_prof}}
\end{figure}

\subsection{Relative efficiency of peeling versus the Hutchinson estimator}\label{sec:hutch}
As discussed in \cref{sec:peel}, a common alternative statistical approach for approximating the trace of a matrix {$\m{G}$} is the estimator of Hutchinson \cite{hutchinson} seen in \eqref{eq:hutch}.  The aim of this section is to show that for matrices with hierarchical low-rank structure our peeling-based algorithm can be much more efficient when a high-accuracy trace approximation is desired.

As in \cref{sec:peelingruntime}, we take our observations to be a regular grid discretizing $[0,100]^2\subset \R^2$ using the Mat\'ern kernel of \eqref{eq:specmatern} with noise $\sigma_N^2 = 1\times 10^{-4}$ and parameter vector $\theta = [10,7]$.  We fix the number of observations at $n=64^2$ and consider how the accuracy of the trace approximation varies with the number of applications of the black-box operator for both weak peeling and the Hutchinson estimator.

Using a high-accuracy recursive skeletonization factorization to construct the black-box operator in \eqref{eq:Gbb} as in \cref{sec:peelingruntime}, we vary the tolerance $\epsilon_\text{peel}$ in the peeling algorithm and plot in \cref{fig:hutch_prof} the relative error in the trace approximation as a function of both the number of black-box applies and total peeling runtime.  Additionally, for the Hutchinson estimator we use the same factorizations to construct the unsymmetric operator $\m{G}' = \m{\Sigma}^{-1}\m{\Sigma}_1$.  We plot the same quantities for a given instantiation of the estimator for comparison.

For low-accuracy approximations with relative error on the order of $1\times 10^{-1}$ to $1\times 10^{-3}$, we see that the Hutchinson estimator is a competitive alternative to the peeling algorithm for finding the trace.  When increased accuracy is desired, however, it is clear that in our examples that the peeling algorithm is the more attractive option.  While the Hutchinson estimator has a simple form and is easy to compute, the relatively slow inverse square root convergence means that $M$ in \eqref{eq:hutch} must be taken to be exceedingly large to drive the variance down to reasonable levels, whereas the peeling algorithm is observed to make more economical use of its black-box matrix-vector products.  It is worth noting that, for this choice of $n$, only 4096 applies are needed to explicitly construct all diagonal entries of the operator via application to the identity, though this is not feasible for larger $n$.

\begin{figure}%
\centering
\begin{minipage}{0.4\textwidth}%
\centering
\includegraphics[width=\textwidth]{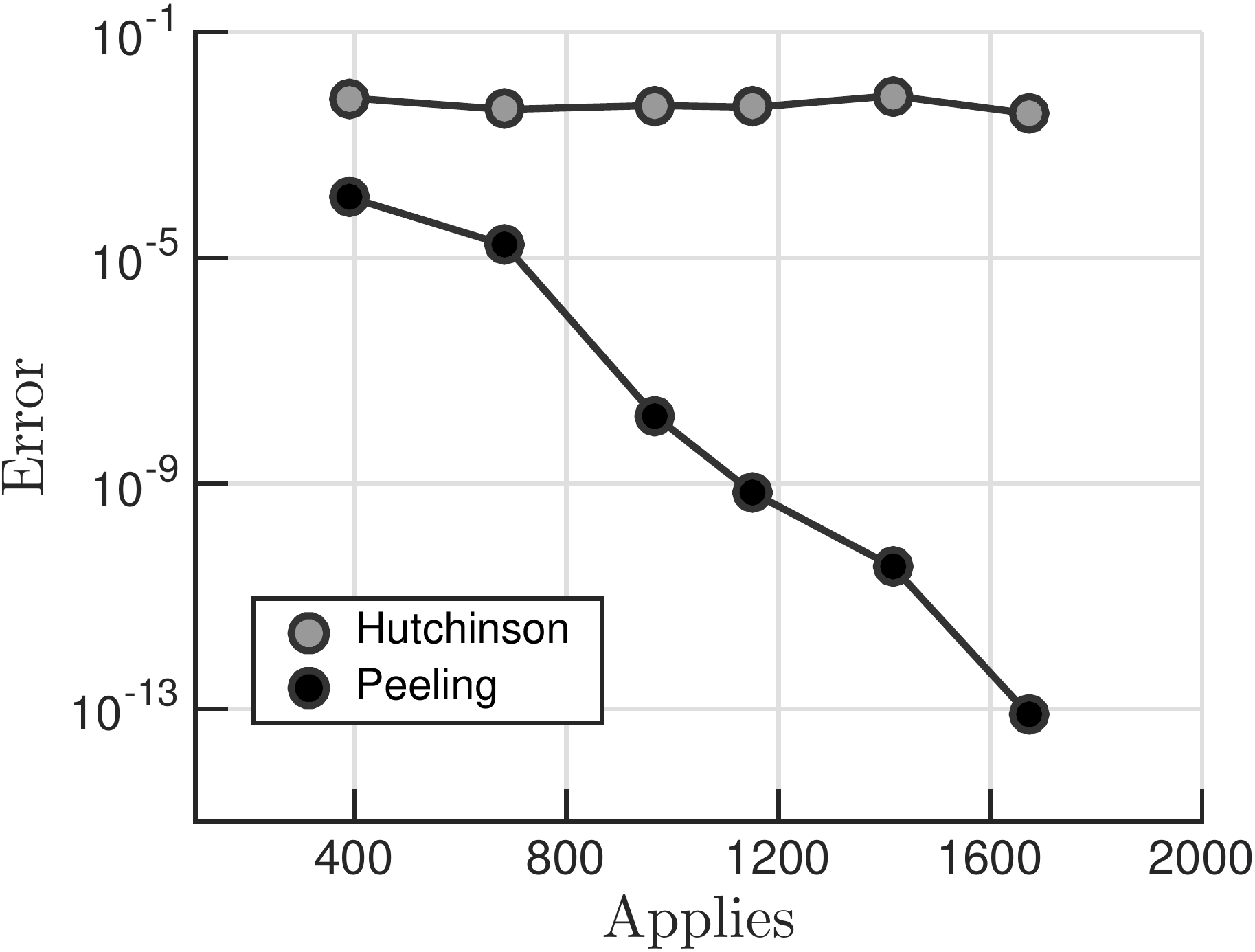}
\end{minipage}%
\quad\quad\quad\quad
\begin{minipage}{0.4\textwidth}%
\centering
\includegraphics[width=\textwidth]{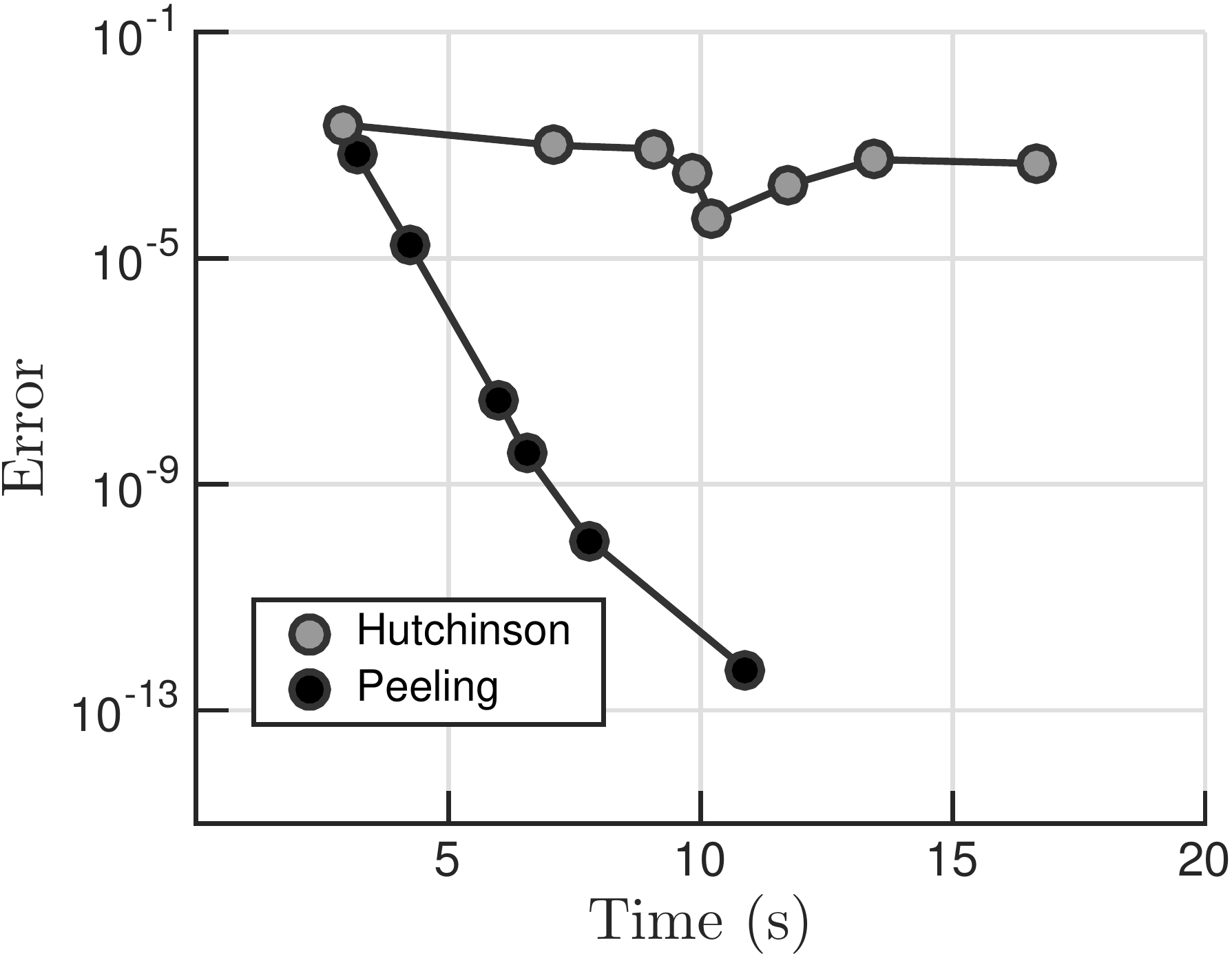}
\end{minipage}%

\caption{Plotting the relative error in the trace approximation versus the number of applications of the black-box operator, we see in the left figure that the Hutchinson estimator exhibits characteristic inverse square root convergence as dictated by the central limit theorem.  In contrast, using the peeling algorithm described in \cref{sec:peelsub}, we see that the same number of black-box applies yields a much improved accuracy, though the rate of convergence depends on the spectra of off-diagonal blocks of the operator.  In the right figure, we plot the error of each method versus wall-clock time to establish that the same scaling behavior holds when error is viewed as a function of time-to-solution{.}  \label{fig:hutch_prof}}
\end{figure}

\subsection{Gridded synthetic data example}\label{sec:uniform}
We now profile a full objective function and gradient evaluation for the MLE problem for $\theta$.  As before, we consider the Mat\'ern kernel of \eqref{eq:specmatern} with noise $\sigma_N^2 = 1\times 10^{-4}$.

We set the parameter vector at $\theta = [10,7]$ and again take the observation locations to be a regular $\sqrt{n}\times\sqrt{n}$ grid discretizing the square $[0,100]^2$.  Evaluating $\ell(\theta)$ and $g_i$ for $i=1,\dots,2$ then requires three skeletonization factorizations and two different trace approximations.  We investigate the algorithm's performance for two different peeling tolerances $\epsilon_\text{peel}$, and in each case take the factorization tolerance to be $\epsilon_\text{fact}=\nobreak\frac{1}{1000}\epsilon_\text{peel}$.  For varying $n$ between $64^2$ and $512^2$, we measured the runtime of both the factorization portion and peeling portion of \cref{alg:mle}.  We note that, given the factorizations and peeled trace estimates, the remainining pieces of \cref{alg:mle} are several orders of magnitude less costly in terms of runtime.

In \cref{fig:it_time} (left), we plot the total runtime for a single objective function and gradient evaluation for the uniform grid of observations (corresponding data in \cref{tab:griddata}).  We see from the figure that the runtime seems to scale as roughly $\O(n^{3/2})$ with the number of observations; a least-squares fit of the data gives $\O(n^{1.6})$.  As can be seen in the table, the amount of time spent in calculating the recursive skeletonization factorizations is roughly an order of magnitude less than the time spent in the peeling trace approximation, and, further, scales slightly better than peeling for this example.

\begin{table}
\caption{Runtime for one objective function and gradient evaluation (\ie, the work for a single iteration) on a uniform grid of observations.\label{tab:griddata}}
\centering
\ra{1.0}
\begin{tabular}{ccccc} \toprule
 $\epsilon_\text{peel}$ & $n$ & $t_\text{fact}$ (s) & $t_\text{peel,weak}$ (s)  & $t_\text{total}$ (s) \\ \hline
 \multirow{ 4}{*}{$1\times 10^{-6}$} &$64^2$ & $7.34\times 10^0$ & $1.23\times 10^1$ & $1.96\times 10^1$ \\ 
 &$128^2$ & $4.86\times 10^1$ & $1.20\times 10^2$ & $1.68\times 10^2$ \\
 &$256^2$ & $2.73\times 10^2$  & $1.22\times 10^3$ & $1.49\times 10^3$\\ 
 &$512^2$& $1.41\times 10^3$ & $1.39\times 10^4$ &$1.53\times 10^4$  \\ \midrule
  \multirow{ 4}{*}{$1\times 10^{-8}$} &$64^2$ & $9.70\times 10^0$ & $1.55\times 10^1$&  $2.52\times 10^1$ \\
 &$128^2$ & $7.08\times 10^1$ & $1.68 \times 10^2$ &$2.39\times10^2$\\ 
 &$256^2$ & $4.28\times 10^2$  & $1.76\times 10^3$ & $2.29\times 10^3$\\
 &$512^2$& $2.64\times 10^3$ &  $1.54\times 10^4$& $1.80\times 10^4$  \\ \bottomrule
\end{tabular}
\end{table}

\begin{figure}%
\centering
\begin{minipage}{0.4\textwidth}%
\centering
\includegraphics[width=\textwidth]{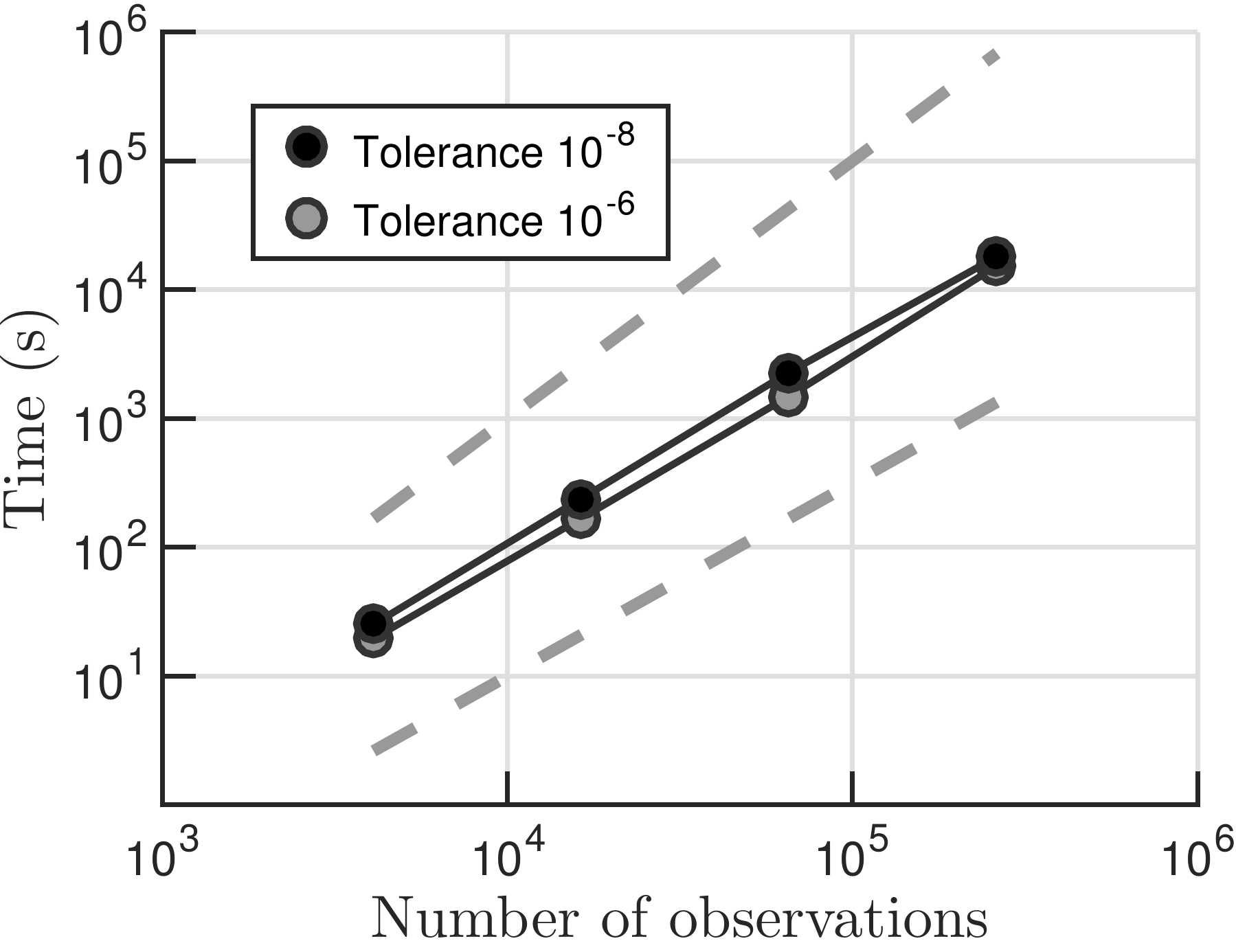}
\end{minipage}%
\quad\quad\quad\quad
\begin{minipage}{0.4\textwidth}%
\centering
\includegraphics[width=\textwidth]{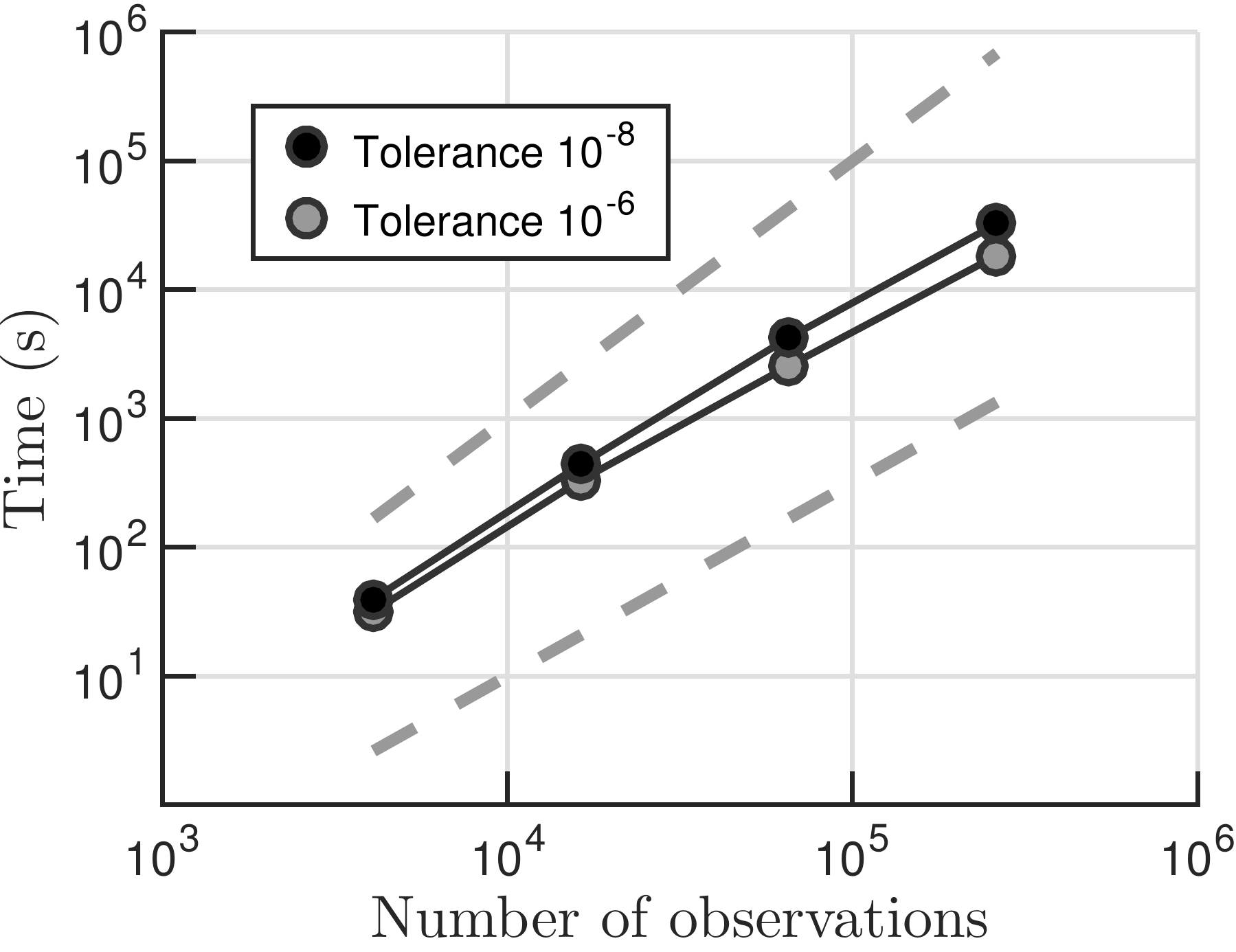}
\end{minipage}%

\caption{On the left we plot the total runtime of evaluating a single objective function and gradient for the uniform grid example of \cref{sec:uniform} as a function of the total number of observations for two different tolerances.  The top trend line shows $\O(n^2)$ scaling and the bottom shows $\O(n^{3/2})$ scaling.  On the right we plot the corresponding results for the scattered data of \cref{sec:ocean1} with the same trend lines. We observe that the scaling in all cases looks like $\O(n^{3/2})$. \label{fig:it_time}}
\end{figure}

\subsection{Scattered synthetic data example}\label{sec:ocean1}
While all examples thus far have used a regular grid of observations, our framework does not rely on this assumption.  To complement the examples on gridded observations, we repeat the same experiment from the previous section with real-world observation locations coming from release 2.5 of the International Comprehensive Ocean-Atmosphere Data Set (ICOADS) \cite{data} obtained from the National Center for Atmospheric Research at \url{http://rda.ucar.edu/datasets/ds540.0/}.  We subselect from ICOADS a set of sea surface temperatures measured at varying locations in the North Atlantic ocean between the years 2008 and 2014.  Restricting the data to observations made in the month of July across all years and obtain roughly 300,000 unique observation locations and corresponding sea surface temperature measurements{, some of which can be seen in \cref{fig:ocean}.} 

{
Because large-scale spatial measurements typically cover a non-trivial range of latitudes and longitudes, the development of valid covariance functions on the entire sphere that respect the proper distance metric has been the subject of much recent work, see, e.g., Gneiting \cite{gneiting2013} and related work \cite{porcu,jun2008,heaton}.  As the focus of this manuscript is not statistical modeling, we employ a simplified model based on Mercator projection of the observations to two spatial dimensions.  Note that the choice of axis scaling in the Mercator projection is arbitrary; in our convention the horizontal axis spans 90 units and the vertical axis spans 70 units.}

To perform scaling tests on the cost of an objective function and gradient evaluation according to \cref{alg:mle}, we subselect from our full dataset of unique observation locations by drawing observations uniformly at random without replacement.  \Cref{fig:it_time} (right) shows the runtime scaling results as a function of the number of observations, with corresponding data in \cref{tab:scattereddata}.  We see that the runtime scaling for the scattered observations follows essentially the same scaling behavior as the gridded observations from \cref{sec:uniform}, with observed complexity between $\O(n^{1.5})$ and $\O(n^{1.6})$.  Again, the skeletonization factorizations take considerably less time than the trace estimation.

\begin{table}
\caption{Runtime for one objective function and gradient evaluation (\ie, the work for a single iteration) on scattered observations with locations from ICOADS.\label{tab:scattereddata}}
\centering
\ra{1.0}
\begin{tabular}{ccccc} \toprule
 $\epsilon_\text{peel}$ & $n$ & $t_\text{fact}$ (s) & $t_\text{peel,weak}$ (s)  & $t_\text{total}$ (s) \\ \midrule
 \multirow{ 4}{*}{$1\times 10^{-6}$} &$2^{12}$ &$7.82\times 10^0$ & $2.41\times 10^1$ & $3.19\times 10^1$ \\
 &$2^{14}$ & $4.25\times10^1$ & $2.91\times 10^2$ & $3.34\times 10^2$ \\
 &$2^{16}$ & $8.60 \times 10^1$  & $2.49\times 10^3$ & $2.57\times10^3$ \\
 &$2^{18}$& $4.84\times 10^2$ & $1.79\times 10^4$& $1.84\times10^4$\\ \midrule
  \multirow{ 4}{*}{$1\times 10^{-8}$} &$2^{12}$ & $1.06\times 10^1$ & $2.92\times 10^1$& $3.98\times 10^1$ \\
 &$2^{14}$ & $5.96\times 10^1$ & $3.83\times 10^2$ &$4.43\times 10^2$\\
 &$2^{16}$ & $2.72\times 10^2$ & $3.98\times 10^3$ & $4.26\times10^3$ \\
 &$2^{18}$& $1.03 \times 10^3$ &  $3.15\times 10^4$& $3.25\times 10^4$  \\ \bottomrule
\end{tabular}
\end{table}

As an illustrative example of the full power of \cref{alg:mle} in context, we take a subset of $n=2^{16}$ scattered observations and realize an instance of a Gaussian process at those locations with true parameter vector $\theta^*=[10,7]$ and noise parameter $\sigma_N^2=1\times10^{-4}$ to generate the observation vector $z$.  Setting the peel tolerance to $\epsilon_\text{peel}=1\times 10^{-6}$ and the factorization tolerance to $\epsilon_\text{fact}=1\times 10^{-9}$, we plugged our approximate log-likelihood and gradient routines into the MATLAB\textsuperscript{\textregistered} routine \texttt{fminunc} for unconstrainted optimization using the quasi-Newton option.  Starting from an initial guess of $\theta_0=[3,30],$ we found that after 13 iterations (14 calls to \cref{alg:mle}) the first-order optimality as measured by the $\ell_\infty$-norm of the gradient had been reduced by three orders of magnitude, yielding an estimate of $\hat\theta=[10.0487,7.0496]$ after approximately $4.86\times 10^4$ seconds.   

\begin{remark}
While a large percentage of this runtime was spent in the peeling algorithm, we find it worthwhile to note that in this example the use of our gradient approximation proved essential---using finite difference approximations to the gradient led to stagnation at the first iteration, even with a factorization tolerance $\epsilon_\text{fact}=1\times 10^{-15}$, \ie, at the limits of machine precision.
\end{remark}  

Because the number of iterations to convergence depends on many factors (\eg, the choice of optimization algorithm, how well the data can be modeled by a Gaussian process, and many convergence tolerances depending on the chosen algorithm), we do not find it useful to attempt to profile the full minimization algorithm more extensively than this, but direct the reader instead to the single-iteration results.

\subsection{{Scattered ocean data example}}\label{sec:ocean2}

  While the factorizations and peeling in \cref{alg:mle} depend only on the locations of the observations and not their values, the log-likelihood $\ell(\cdot)$ can have a more complicated shape with real observations $z$ than with synthetic data, which may impact the required tolerance parameters $\epsilon_\text{fact}$ and $\epsilon_\text{peel}$ and the difficulty of MLE.  Further, there are a number of practical considerations relevant for real data not addressed thus far in our synthetic examples.

  As a refinement of \eqref{eq:normaldist}, suppose now that the data are distributed according to $z \sim N(\mu{\bf 1}, \sigma^2\m{\Sigma}(\theta))$, where ${\bf 1}\in\R^n$ is the all-ones vector, $\mu$ and $\sigma^2$ represent the constant but unknown mean and variance level, and our parameterized Mat\'ern model is given by (for several different $\nu$)
  {
\begin{align}\label{eq:diffkernels}
[\m{\Sigma}(\theta)]_{ij} &= \frac{1}{\Gamma(\nu)2^{\nu-1}}\left(\frac{\sqrt{2\nu}r_{ij}}{\rho}\right)^\nu K_\nu\left(\frac{\sqrt{2\nu}r_{ij}}{\rho}\right) + \sigma_N^2 \delta_{ij}\\
&= \left\{\begin{array}{ll}
\exp\left(-\frac{r_{ij}}{\rho}\right),&\nu = 1/2,\\
 \left(1 + \frac{\sqrt{3}r_{ij}}{\rho}\right)\exp\left(-\frac{\sqrt{3}r_{ij}}{\rho}\right) + \sigma_N^2 \delta_{ij}, &\nu=3/2,\\ 
 \left(1 + \frac{\sqrt{5}r_{ij}}{\rho} + \frac{5r_{ij}^2}{3\rho^2}\right)\exp\left(-\frac{\sqrt{5}r_{ij}}{\rho}\right) + \sigma_N^2 \delta_{ij},
 &\nu=5/2, \end{array} \right.\nonumber
\end{align}}
with $r_{ij}=\|x_i-x_j\|$.
In this example, the parameter vector is $\theta=[{\rho},\sigma_N^2]$, consisting of a {single correlation length parameter} and the noise level.  To optimize the new log-likelihood over $\mu$, $\sigma^2$, and $\theta$, we note that optimization over $\mu$ and $\sigma^2$ results in closed form expressions for these parameters in terms of $\theta$, which may then be substituted back into \cref{eq:loglikelihood} to obtain the \emph{log profile likelihood} for this model
  \begin{align}\label{eq:prof}
  \tilde\ell(\theta) &\equiv  -\frac{1}{2}\log |\m{\Sigma}| - \frac{n}{2} \log \left(z^T(\m{\Sigma} + {\bf 1}{\bf 1}^T)^{-1}z\right) +\frac{n}{2}\left(\log n - 1 - 2\pi\right) 
  \end{align}
  with gradient components given by
\begin{align*}\label{eq:gradientprof}
\tilde g_i&\equiv -\frac{1}{2} \Tr (\m{\Sigma}^{-1}\m{\Sigma}_i) + \frac{n}{2} \left(\frac{z^T(\m{\Sigma} + {\bf 1}{\bf 1}^T)^{-1}\m{\Sigma}_i(\m{\Sigma} + {\bf 1}{\bf 1}^T)^{-1}z}{z^T(\m{\Sigma} + {\bf 1}{\bf 1}^T)^{-1}z}\right)
,\quad i=1,\dots,p.
\end{align*}
 Optimization of this new model fits neatly into the computational framework of \cref{alg:mle} with trivial modifications.  The new model has the advantage of greater plausibility, though it still admits many further improvements.

\begin{table}
\caption{Parameter estimates $\hat\theta = [{\hat \rho}, \widehat{\sigma_N^2}]$, the corresponding mean and variance level $\mu(\hat\theta)$ and $\sigma^2(\hat\theta)$, and the log-likelihood values of the fitted parameters for the model with kernel \eqref{eq:diffkernels}. We note that for $\nu=1/2$, the estimate of $\sigma_N^2$ is at its lower bound.  \label{tab:proflikelihood}}
\centering
\ra{1.0}
\begin{tabular}{cccccc} \toprule
$\nu$ & ${\hat \rho}$ & $\widehat{\sigma_N^2}$  & $\mu(\hat\theta)$ & $\sigma^2(\hat\theta)$ & $ \ell(\hat\theta)$\\ \midrule
 $1/2$& $ {2.12\times10^{+0}}   $ & ${1.00\times 10^{-8}}$ & ${-5.65\times10^{-1}}$ & ${9.97\times 10^{-1}}$& ${9.28\times 10^{+4}}$ \\ 
 $3/2$& $ {2.52\times 10^{-1}}$ & ${4.37\times 10^{-3}}$ & $   {-4.38\times10^{-1}}$ & ${8.71\times 10^{-1}}$& ${8.42\times 10^{+4}}$  \\ 
 $5/2$& $ {1.78\times 10^{-1}}    $ & {$5.67\times 10^{-3}$} & $ {-4.30\times 10^{-1}}$ & ${8.57\times 10^{-1}}$& ${7.74\times 10^{+4}}$ \\\bottomrule
\end{tabular}
\end{table}

From the full set of sea surface temperature observations, we subselected $n=2^{17}$ unique temperature measurements corresponding to observations between July 2013 and August 2013.    Taking $\epsilon_\text{fact}=1\times10^{-9}$ and {$\epsilon_\text{peel}=1\times10^{-6}$}, we use the MATLAB\textsuperscript{\textregistered} optimization routine $\texttt{fmincon}$ with the `SQP' option to estimate the {correlation length} parameter $\theta_1={\rho}$ and noise parameter $\theta_2=\sigma_N^2$ for the standardized temperature measurements.  {For $\nu=3/2$ and $\nu=5/2$ this was accomplished by numerically maximizing \eqref{eq:prof} subject to the lower-bound constraint $\sigma_N^2\ge 1\times 10^{-5}$, which was necessary to ensure $\m{\Sigma}$ was not numerically rank-deficient.  For $\nu=1/2$ the covariance matrix $\m{\Sigma}$ is naturally better conditioned so a looser lower-bound $\sigma_N^2 \ge 1\times 10^{-8}$ was used.}  The tolerances dictating the minimum step-size and minimum change in objective function between successive iterates were both set to {$1\times 10^{-6}$}.  

Due to the non-convex nature of the problem, we tried several choices of starting parameter {for each $\nu$}; the results we present are for the best initialization in each case.  {For choices of initial parameters leading to convergent iterates (\eg, $\theta_0=[5,1]$ or $\theta_0=[1\times 10^{-1},1\times 10^{-3}]$), the converged solutions all agreed to the specified tolerance and the objective function value at the optimal points agreed to six digits.  For some choices of initial parameters, the optimization terminated prematurely due to the relative improvement tolerances used to evaluate convergence (\ie, when the initial parameters are very poor, even a large improvement relative to the initial parameters can be far from the best choice of parameters).  We did not observe any evidence of multiple local optima, though the possibility  that our reported parameters are globally suboptimal cannot be ruled out.}

The results of our numerical optimization for each choice of $\nu$ can be seen in \cref{tab:proflikelihood}, where in each case optimization terminated due to the step-size tolerance.   At each corresponding $\hat\theta$, however, we note that the gradient is small relative to the objective function.

Of the three different models, we find that the fitted model for $\nu=1/2$ gives the best fit as measured both by comparative likelihood and qualitatively (see \cref{fig:ocean}).  Since $\nu$ dictates the smoothness of denoised process, these results imply that the best description of the observed data among our choices is the one with the least assumptions on smoothness.  We caution that this does not preclude a much better fit with a more sophisticated model, but this simple example illustrates that our framework is effective for MLE even for real observations.

\begin{figure}
\centering
\includegraphics[width=\textwidth]{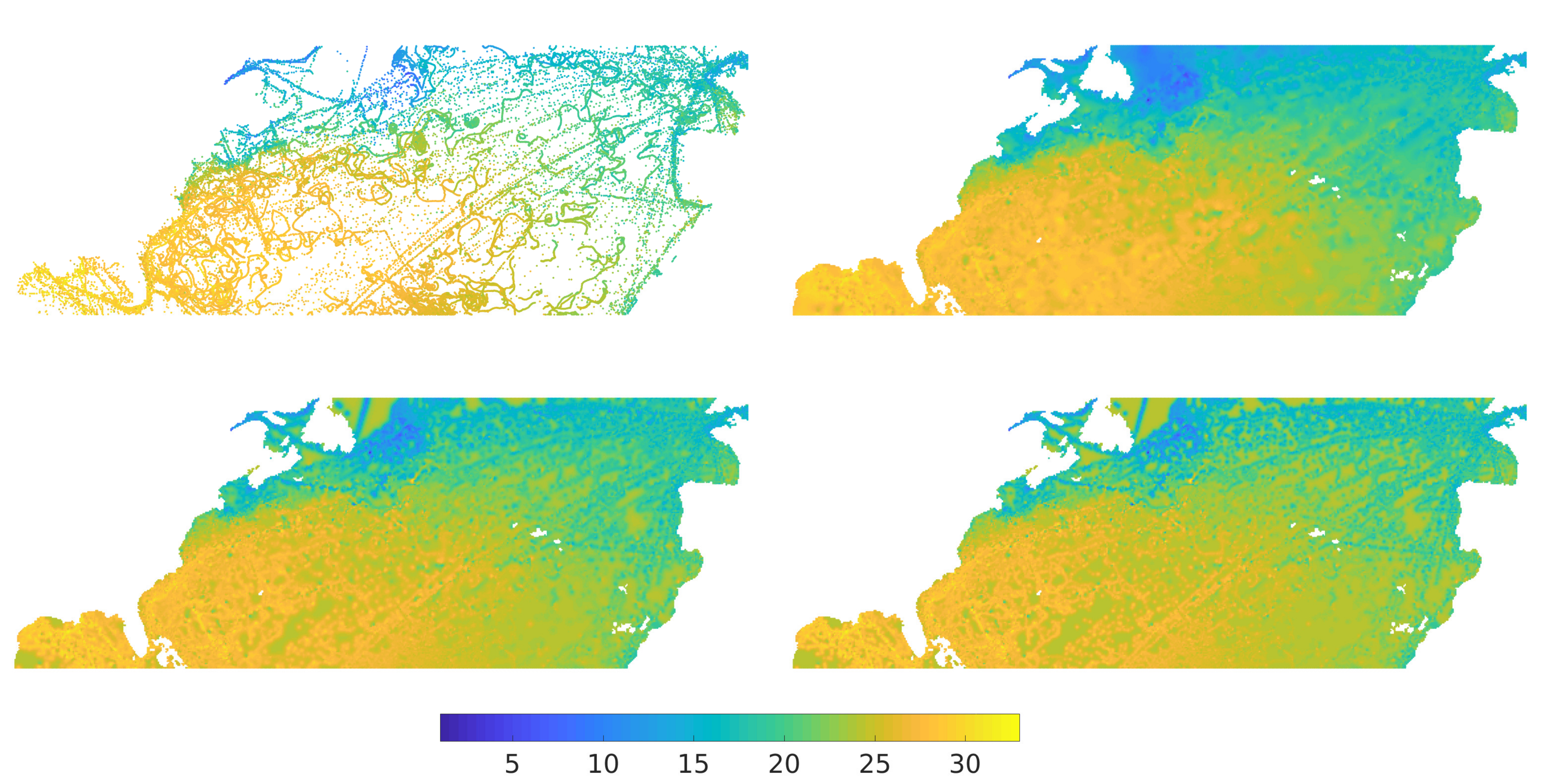}
\caption{In the top-left plot, we show a subselection of $n=2^{17}$ Atlantic ocean surface temperature measurements from ICOADS projected to a 2D plane through Mercator projection and then scattered on top of a white background for visualization.  Fitting the model with constant mean and covariance given by \eqref{eq:diffkernels}, we use the estimated parameters to find the conditional mean temperatures throughout this region of the Atlantic for $\nu=1/2$ (top right), $\nu=3/2$ (bottom left), and $\nu=5/2$ (bottom right).  The color bar shows the estimated sea surface temperature in Celsius.  \label{fig:ocean}}
\end{figure}

\section{Conclusions}\label{sec:conclusion}
The framework for Gaussian process MLE presented in this paper and summarized in \cref{alg:mle} provides a straightforward method of leveraging hierarchical matrix representations from scientific computing for fast computations with kernelized covariance matrices arising in spatial statistics.  The general linear algebraic approach to approximating off-diagonal blocks of the covariance matrix to a specified error tolerance by adaptively determining their ranks gives a flexible way of attaining high-accuracy approximations with reasonable runtime.  A further merit to this approach is that it does not rely on having gridded observations or a translation-invariant covariance kernel.

While in this paper we have focused on {maximum likelihood} estimation for Gaussian processes, these methods are equally viable for {the Bayesian setting.  For example, computing maximum \emph{a posteriori} estimates follows essentially the same approach with the addition of a term depending on the prior.  Further, sampling from the posterior distriution of $\theta$ in a Bayesian setting can be accomplished using standard Markov chain Monte Carlo methods based on quickly evaluating the likelihood and posterior.  This can also be combined with \cref{remark:conditional} for a fully Bayesian treatment.}

Our numerical results in \cref{sec:results} show that our framework scales favorably when applied to our two test cases {(the rational quadratic and Mat\'ern family kernels)}, leading to runtimes scaling approximately as $\O(n^{3/2})$ with $n$ the number of observations.  Further, we see that the tolerance parameter $\epsilon_\text{peel}$ controlling the rank of off-diagonal block approximations in the peeling algorithm serves as a good estimate of the order of the error in the ultimate trace approximation as well.  In practice, the tolerances $\epsilon_\text{peel}$ and $\epsilon_\text{fact}$ can be dynamically modified during the course of the maximum likelihood process for performance, \eg, one could use relatively low-accuracy approximations during initial iterations of the {optimization} routine and slowly decrease the tolerance as the optimization progresses.

While the methods and complexity estimates discussed in this paper relate to the case of two spatial dimensions, they trivially extend to one-dimensional (time-series) data or quasi-two-dimensional data, \eg, observations in three dimensions where the sampling density in one dimension is much smaller than in the other two.  While the same methods apply in principle to truly three-dimensional data, the corresponding computational complexity is bottlenecked by the cost of using peeling to obtain a high-accuracy trace estimate of the matrices $\m{\Sigma}^{-1}\m{\Sigma}_i$ for $i=1,\dots,p$ due to increased rank growth.  In fact, even in the two-dimensional case it is clear from \cref{tab:griddata} and \cref{tab:scattereddata} that the most expensive piece of of our framework in practice is determining these traces.  One solution is to instead use the hierarchical matrix representations inside of an estimator such as that of Stein \etal \cite{Stein2013}, which obviates the need for the trace.  For the true MLE, however, future work on efficiently computing this trace to high accuracy is necessary.  {Given a method for efficiently computing this trace for 3D data, we expect that related {factorizations based on more sophisticated use of skeletonization} should give complexities for computing the log-likelihood and gradient that are as good as or better than those we obtain with recursive skeletonization in the 2D case.  For example, the hierarchical interpolative factorization \cite{hifie} {(which uses further levels of compression to mitigate rank-growth of off-diagonal blocks) may be used in our framework  as an efficient method of applying} $\m{\Sigma}_i$ and $\m{\Sigma}^{-1}$ and computing the log-determinant of $\m{\Sigma}$ for 3D problems.}

{While Gaussian process regression is widely used for data in $\R^d$ with $d$ much larger than three, the methods of this paper are designed with spatial data in mind.  In particular, in the high-dimensional setting the geometry of the observations becomes very important for efficiency.  If the data can be well-approximated according to an intrinsic low-dimensional embedding that is efficient to identify, there is hope for efficient approximations using hierarchical rank structure (see, for example, Yu et al.\ \cite{yu}).  However, in general we expect that rank-structured factorizations will continue to be most effective for low-dimensional spatial applications. }

\section*{Acknowledgments}
The authors thank Matthias Cremon, Eileen Martin, Sven Schmit, and Austin Benson for useful discussion on Gaussian process regression, the anonymous reviewers for thoughtful comments that improved the presentation of this paper, and Stanford University and the Stanford Research Computing Center for providing computational resources and support that have contributed to these research results.

\bibliographystyle{siamplain}
\bibliography{M111647}

\begin{thebibliography}{10}

\bibitem{ambikasaran}
{\sc S.~Ambikasaran, D.~Foreman-Mackey, L.~Greengard, D.~W. Hogg, and
  M.~O'Neil}, {\em Fast direct methods for {G}aussian processes}, IEEE
  Transactions on Pattern Analysis and Machine Intelligence, 38 (2016),
  pp.~252--265.

\bibitem{Anitescu2012}
{\sc M.~Anitescu, J.~Chen, and L.~Wang}, {\em A matrix-free approach for
  solving the parametric {G}aussian process maximum likelihood problem}, SIAM
  Journal on Scientific Computing, 34 (2012), pp.~A240--A262.

\bibitem{aune}
{\sc E.~Aune, D.~P. Simpson, and J.~Eidsvik}, {\em Parameter estimation in high
  dimensional {G}aussian distributions}, Statistics and Computing, 24 (2014),
  pp.~247--263, \href{http://dx.doi.org/10.1007/s11222-012-9368-y}
  {doi:10.1007/s11222-012-9368-y},
  \url{http://dx.doi.org/10.1007/s11222-012-9368-y}.

\bibitem{borm2007}
{\sc S.~B{\"o}rm and J.~Garcke}, {\em Approximating {G}aussian processes with
  $\mathcal{H}^2$-matrices}, in Proceedings of the 18th European Conference on
  Machine Learning, Springer, 2007, pp.~42--53.

\bibitem{castrillon}
{\sc J.~E. Castrill\'on-Cand\'as, M.~G. Genton, and R.~Yokota}, {\em
  Multi-level restricted maximum likelihood covariance estimation and {K}riging
  for large non-gridded spatial datasets}, Spatial Statistics,  (2015).

\bibitem{fasthss}
{\sc S.~Chandrasekaran, P.~Dewilde, M.~Gu, W.~Lyons, and T.~Pals}, {\em A fast
  solver for {HSS} representations via sparse matrices}, SIAM Journal on Matrix
  Analysis and Applications, 29 (2007), pp.~67--81.

\bibitem{fastulv}
{\sc S.~Chandrasekaran, M.~Gu, and T.~Pals}, {\em A fast {ULV} decomposition
  solver for hierarchically semiseparable representations}, SIAM Journal on
  Matrix Analysis and Applications, 28 (2006), pp.~603--622.

\bibitem{id}
{\sc H.~Cheng, Z.~Gimbutas, P.-G. Martinsson, and V.~Rokhlin}, {\em On the
  compression of low rank matrices}, SIAM Journal on Scientific Computing, 26
  (2005), pp.~1389--1404.

\bibitem{corona2013}
{\sc E.~Corona, P.-G. Martinsson, and D.~Zorin}, {\em {An $O(N)$ direct solver
  for integral equations on the plane}}, Applied and Computational Harmonic
  Analysis, 38 (2015), pp.~284 -- 317.

\bibitem{cressie}
{\sc N.~Cressie and G.~Johannesson}, {\em Fixed rank {K}riging for very large
  spatial data sets}, Journal of the Royal Statistical Society, Series B, 70
  (2008), pp.~209--226.

\bibitem{eidsvik}
{\sc J.~Eidsvik, B.~A. Shaby, B.~J. Reich, M.~Wheeler, and J.~Niemi}, {\em
  Estimation and prediction in spatial models with block composite
  likelihoods}, Journal of Computational and Graphical Statistics, 23 (2014),
  pp.~295--315, \href{http://dx.doi.org/10.1080/10618600.2012.760460}
  {doi:10.1080/10618600.2012.760460},
  \url{http://dx.doi.org/10.1080/10618600.2012.760460},
  \href{http://arxiv.org/abs/http://dx.doi.org/10.1080/10618600.2012.760460}
  {arXiv:http://dx.doi.org/10.1080/10618600.2012.760460}.

\bibitem{nychka}
{\sc R.~Furrer, M.~G. Genton, and D.~Nychka}, {\em Covariance tapering for
  interpolation of large spatial datasets}, Journal of Computational and
  Graphical Statistics, 15 (2006), pp.~502--523.

\bibitem{domainsAd}
{\sc A.~Gillman, P.~M. Young, and P.-G. Martinsson}, {\em A direct solver with
  ${O(N)}$ complexity for integral equations on one-dimensional domains},
  Frontiers of Mathematics in China, 7 (2012), pp.~217--247.

\bibitem{gneiting2013}
{\sc T.~Gneiting}, {\em Strictly and non-strictly positive definite functions
  on spheres}, Bernoulli, 19 (2013), pp.~1327--1349,
  \href{http://dx.doi.org/10.3150/12-BEJSP06} {doi:10.3150/12-BEJSP06},
  \url{http://dx.doi.org/10.3150/12-BEJSP06}.

\bibitem{gg}
{\sc L.~Greengard, D.~Gueyffier, P.-G. Martinsson, and V.~Rokhlin}, {\em Fast
  direct solvers for integral equations in complex three-dimensional domains},
  Acta Numerica, 18 (2009), pp.~243--275.

\bibitem{greengard-rokhlin}
{\sc L.~Greengard and V.~Rokhlin}, {\em On the numerical solution of two-point
  boundary value problems}, Communications on Pure and Applied Mathematics, 44
  (1991), pp.~419--452.

\bibitem{Hackbusch}
{\sc W.~Hackbusch}, {\em A sparse matrix arithmetic based on
  $\mathcal{H}$-matrices. {P}art {I}: Introduction to $\mathcal{H}$-matrices},
  Computing, 62 (1999), pp.~89--108.

\bibitem{hackbook}
{\sc W.~Hackbusch}, {\em Hierarchical Matrices: Algorithms and Analysis},
  Springer Series in Computational Mathematics, Springer-Verlag Berlin
  Heidelberg, 2015.

\bibitem{HackbuschB}
{\sc W.~Hackbusch and S.~B\"orm}, {\em Data-sparse approximation by adaptive
  $\mathcal{H}^2$-matrices}, Computing, 69 (2002), pp.~1--35.

\bibitem{HackbuschK}
{\sc W.~Hackbusch and B.~N. Khoromskij}, {\em A sparse $\mathcal{H}$-matrix
  arithmetic. {P}art {II}: Application to multi-dimensional problems},
  Computing, 64 (2000), pp.~21--47.

\bibitem{Hackbuschweak}
{\sc W.~Hackbusch, B.~N. Khoromskij, and R.~Kriemann}, {\em Hierarchical
  matrices based on a weak admissibility criterion}, Computing, 73 (2004),
  pp.~207--243.

\bibitem{halko}
{\sc N.~Halko, P.-G. Martinsson, and J.~A. Tropp}, {\em Finding structure with
  randomness: Probabilistic algorithms for constructing approximate matrix
  decompositions}, SIAM Review, 53 (2011), pp.~217--288.

\bibitem{rskel}
{\sc K.~L. Ho and L.~Greengard}, {\em A fast direct solver for structured
  linear systems by recursive skeletonization}, SIAM Journal on Scientific
  Computing, 34 (2012), pp.~A2507--A2532.

\bibitem{hifie}
{\sc K.~L. Ho and L.~Ying}, {\em Hierarchical interpolative factorization for
  elliptic operators: Integral equations}, Communications on Pure and Applied
  Mathematics,  (2015).

\bibitem{hutchinson}
{\sc M.~F. Hutchinson}, {\em A stochastic estimator of the trace of the
  influence matrix for {L}aplacian smoothing splines}, Communications in
  Statistics - Simulation and Computation, 19 (1990), pp.~433--450.

\bibitem{jun2008}
{\sc M.~Jun and M.~L. Stein}, {\em Nonstationary covariance models for global
  data}, Ann. Appl. Stat., 2 (2008), pp.~1271--1289,
  \href{http://dx.doi.org/10.1214/08-AOAS183} {doi:10.1214/08-AOAS183},
  \url{http://dx.doi.org/10.1214/08-AOAS183}.

\bibitem{khoromskij2008}
{\sc B.~N. Khoromskij, A.~Litvinenko, and H.~G. Matthies}, {\em Application of
  hierarchical matrices for computing the {K}arhunen--{L}o{\`e}ve expansion},
  Computing, 84 (2008), pp.~49--67.

\bibitem{Lin2011}
{\sc L.~Lin, J.~Lu, and L.~Ying}, {\em {Fast construction of hierarchical
  matrix representation from matrix-vector multiplication}}, Journal of
  Computational Physics, 230 (2011), pp.~4071--4087,
  \href{http://arxiv.org/abs/1001.0149} {arXiv:1001.0149}.

\bibitem{lindgren}
{\sc F.~Lindgren, H.~Rue, and J.~Lindstr{\"o}m}, {\em An explicit link between
  {G}aussian fields and {G}aussian {M}arkov random fields: the stochastic
  partial differential equation approach}, Journal of the Royal Statistical
  Society: Series B (Statistical Methodology), 73 (2011), pp.~423--498,
  \href{http://dx.doi.org/10.1111/j.1467-9868.2011.00777.x}
  {doi:10.1111/j.1467-9868.2011.00777.x},
  \url{http://dx.doi.org/10.1111/j.1467-9868.2011.00777.x}.

\bibitem{martinsson-rokhlin}
{\sc P.-G. Martinsson and V.~Rokhlin}, {\em A fast direct solver for boundary
  integral equations in two dimensions}, Journal of Computational Physics, 205
  (2005), pp.~1--23.

\bibitem{matheron}
{\sc G.~Matheron}, {\em Principles of geostatistics}, Economic geology, 58
  (1963), pp.~1246--1266.

\bibitem{rss}
{\sc V.~Minden, K.~L. Ho, A.~Damle, and L.~Ying}, {\em A recursive
  skeletonization factorization based on strong admissibility}, Multiscale
  Modeling \& Simulation, 15 (2017), pp.~768--796,
  \href{http://dx.doi.org/10.1137/16M1095949} {doi:10.1137/16M1095949},
  \url{http://dx.doi.org/10.1137/16M1095949},
  \href{http://arxiv.org/abs/http://dx.doi.org/10.1137/16M1095949}
  {arXiv:http://dx.doi.org/10.1137/16M1095949}.

\bibitem{porcu}
{\sc E.~Porcu, M.~Bevilacqua, and M.~G. Genton}, {\em Spatio-temporal
  covariance and cross-covariance functions of the great circle distance on a
  sphere}, Journal of the American Statistical Association, 111 (2016),
  pp.~888--898, \href{http://dx.doi.org/10.1080/01621459.2015.1072541}
  {doi:10.1080/01621459.2015.1072541},
  \url{http://dx.doi.org/10.1080/01621459.2015.1072541},
  \href{http://arxiv.org/abs/http://dx.doi.org/10.1080/01621459.2015.1072541}
  {arXiv:http://dx.doi.org/10.1080/01621459.2015.1072541}.

\bibitem{rokhlin-tygert}
{\sc V.~Rokhlin and M.~Tygert}, {\em A fast randomized algorithm for
  overdetermined linear least-squares regression}, Proceedings of the National
  Academy of Sciences, 105 (2008), pp.~13212--13217.

\bibitem{Sang2012}
{\sc H.~Sang and J.~Z. Huang}, {\em {A full-scale approximation of covariance
  functions for large spatial data sets}}, Journal of the Royal Statistical
  Society, Series B, 74 (2012), pp.~111--132.

\bibitem{starr_rokhlin}
{\sc P.~Starr and V.~Rokhlin}, {\em On the numerical solution of two-point
  boundary value problems {II}}, Communications on Pure and Applied
  Mathematics, 47 (1994), pp.~1117--1159.

\bibitem{stein1999interpolation}
{\sc M.~L. Stein}, {\em Interpolation of Spatial Data: Some Theory for
  {K}riging}, Springer Series in Statistics, Springer New York, 1999.

\bibitem{steinpc}
{\sc M.~L. Stein, J.~Chen, and M.~Anitescu}, {\em Difference filter
  preconditioning for large covariance matrices}, SIAM Journal on Matrix
  Analysis and Applications, 33 (2012), pp.~52--72.

\bibitem{Stein2013}
{\sc M.~L. Stein, J.~Chen, and M.~Anitescu}, {\em {Stochastic approximation of
  score functions for Gaussian processes}}, Annals of Applied Statistics, 7
  (2013), pp.~1162--1191, \href{http://arxiv.org/abs/1312.2687}
  {arXiv:1312.2687}.

\bibitem{steincomp}
{\sc M.~L. Stein, Z.~Chi, and L.~J. Welty}, {\em Approximating likelihoods for
  large spatial data sets}, Journal of the Royal Statistical Society, Series B,
  66 (2004), pp.~275--296.

\bibitem{tropp}
{\sc J.~A. Tropp}, {\em Improved analysis of the subsampled randomized
  {H}adamard transform}, Advances in Adaptive Data Analysis, 03 (2011),
  pp.~115--126.

\bibitem{vanhatalo}
{\sc J.~Vanhatalo, V.~Pietil{\"a}inen, and A.~Vehtari}, {\em Approximate
  inference for disease mapping with sparse {G}aussian processes}, Statistics
  in Medicine, 29 (2010), pp.~1580--1607.

\bibitem{vecchia}
{\sc A.~V. Vecchia}, {\em Estimation and model identification for continuous
  spatial processes}, Journal of the Royal Statistical Society, Series B, 50
  (1988), pp.~pp. 297--312.

\bibitem{whittle}
{\sc P.~Whittle}, {\em On stationary processes in the plane}, Biometrika, 41
  (1954), pp.~pp. 434--449.

\bibitem{data}
{\sc S.~D. Woodruff, S.~J. Worley, S.~J. Lubker, Z.~Ji, J.~E. Freeman, D.~I.
  Berry, P.~Brohan, E.~C. Kent, R.~W. Reynolds, S.~R. Smith, and C.~Wilkinson},
  {\em {ICOADS} release 2.5: extensions and enhancements to the surface marine
  meteorological archive}, International Journal of Climatology, 31 (2011),
  pp.~951--967.

\bibitem{xia}
{\sc J.~Xia, S.~Chandrasekaran, M.~Gu, and X.~S. Li}, {\em Fast algorithms for
  hierarchically semiseparable matrices}, Numerical Linear Algebra With
  Applications, 17 (2010), p.~953{\textendash}976.

\bibitem{yu}
{\sc C.~D. Yu, W.~B. March, and G.~Biros}, {\em An n log n parallel fast direct
  solver for kernel matrices}, in 2017 IEEE International Parallel and
  Distributed Processing Symposium (IPDPS), May 2017, pp.~886--896,
  \href{http://dx.doi.org/10.1109/IPDPS.2017.10} {doi:10.1109/IPDPS.2017.10}.

\bibitem{heaton}
{\sc M.~ Heaton, M.~Katzfuss, C.~Berrett, and D.~ Nychka}, {\em
  Constructing valid spatial processes on the sphere using kernel
  convolutions}, Environmetrics, 25 (2014), pp.~2--15,
  \href{http://dx.doi.org/10.1002/env.2251} {doi:10.1002/env.2251},
  \url{http://dx.doi.org/10.1002/env.2251}.

\end{thebibliography}

\end{document}